\documentclass[a4paper,USenglish]{lipics-v2021}
\usepackage[hideChapterNumber,refchapterCite,noCCBy]{craigstyle}

\usepackage{craigstyle}

\usepackage[utf8]{inputenc}
\usepackage{qtree}
\usepackage{color}
\usepackage{graphicx}
\usepackage{contour}
\usepackage{caption}
\usepackage{subcaption}
\usepackage{stmaryrd}
\usepackage{paralist}
\usepackage{cmap}           
\usepackage{xspace}

\DeclareMathAlphabet{\mathcal}{OMS}{cmsy}{m}{n}

\usepackage{amsthm}
\usepackage{amsmath}
\usepackage{amssymb}
\usepackage{proof}
\usepackage{xspace}
\usepackage{cellspace}
\usepackage{tikz}
\usetikzlibrary{shapes,backgrounds}

\usepackage{balance}
\usepackage[linesnumbered,vlined,ruled]{algorithm2e}


\newcolumntype{C}[1]{>{\centering\let\newline\\\arraybackslash\hspace{0pt}}m{#1}}
\newcolumntype{L}[1]{>{\let\newline\\\arraybackslash\hspace{0pt}}m{#1}}

\newenvironment{myexmp}{\begin{example}}{\null\hfill\lipicsEnd\end{example}}

\newcommand{\rqfo}{\kw{RQFO}}

\newcommand{\arity}{\kw{arity}}

\newcommand{\adom}{\kw{adom}}

\newcommand{\return}{\kw{Return}}

\newcommand{\indomain}{\kw{InDomain}}
\newcommand{\inmap}{\kw{InMap}}
\newcommand{\outmap}{\kw{OutMap}}

\newcommand{\infacccon}{\kw{InfAccCopy}}

\newcommand{\IS}{induced-subinstance}

\newcommand{\abind}{\kw{vbind}}
\newcommand{\accbind}{\kw{AccBind}}
\newcommand{\smith}{\text{``$\kw{Smith}$''}}
\newcommand{\jones}{\text{``$\kw{Jones}$''}}
\newcommand{\mathematics}{\kw{mathematics}}

\newcommand{\wrt}{w.r.t.}
\newcommand{\ie}{i.e.}
\newcommand{\eg}{e.g.}

\newcommand{\asch}{\kw{Sch}}
\newcommand{\aschema}{\asch}

\newcommand{\posexistsineq}{\exists^{+, \neq}}

\newcommand{\existsineq}{\exists^{\neq}}
\newcommand{\deacc}{\kw{DeAcc}}
\newcommand{\lostarski}{\L{}os-Tarski}

\newcommand{\bindpat}{\kw{BindPatt}}

\newcommand{\nneg}{\mathop{\sim}}

\newcommand{\T}{{\cal T}}

\renewcommand{\S}{{\cal S}}

\newcommand{\true}{\kw{True}}
\newcommand{\false}{\kw{False}}
\newcommand{\select}{\sigma}

\newcommand{\join}{\bowtie}

\newcommand{\accpart}{\kw{AccPart}}
\newcommand{\accs}{\kw{AcSch}}
\newcommand{\accsb}{\kw{AcSch}^\leftrightarrow}
\newcommand{\accsbp}{\kw{AltAcSch}^\leftrightarrow}

\newcommand{\False}{\false}
\newcommand{\True}{\true}
\newcommand{\accsneg}{\kw{AcSch}^\neg}

\newcommand{\spj}{\ensuremath{{SPJ}}^{\neq}\xspace}

\newcommand{\uspj}{\ensuremath{{USPJ}}^{\neq}\xspace}

\newcommand{\uspjneg}{\ensuremath{{USPJAD}^{\neq}}\xspace}
\newcommand{\ESPJ}{\ensuremath{{SPJ}}\xspace}
\newcommand{\espj}{\ESPJ}
\newcommand{\EUSPJ}{\ensuremath{{USPJ}}\xspace}
\newcommand{\euspj}{\EUSPJ}
\newcommand{\EUSPJneg}{\ensuremath{{USPJAD}}\xspace}
\newcommand{\euspjneg}{\EUSPJneg}

\newcommand{\kw}[1]{{\mathsf{#1}}\xspace}
\newcommand{\department}{\kw{Department}}

\newcommand{\uemployee}{\kw{UEmployee}}
\newcommand{\employee}{\kw{Employee}}

\newcommand{\deptid}{\kw{deptid}}
\newcommand{\mgrid}{\kw{mgrid}}

\newcommand{\did}{\deptid}

\newcommand{\profname}{\kw{profname}}
\newcommand{\studentid}{\kw{studid}}
\newcommand{\studid}{\studentid}
\newcommand{\dname}{\kw{dname}}
\newcommand{\profid}{\kw{profid}}
\newcommand{\lname}{\kw{lname}}
\newcommand{\student}{\kw{Student}}
\newcommand{\professor}{\kw{Professor}}

\newcommand{\profinfo}{\kw{Profinfo}}

\newcommand{\univdirectory}{\kw{Udirectory}}
\newcommand{\udirectory}{\univdirectory}

\newcommand{\eid}{\kw{eid}}

\newcommand{\lastname}{\kw{lastname}}

\newcommand{\accessible}{\kw{accessible}}

\newcommand{\acc}[1]{\kw{Accessed} #1}
\newcommand{\accq}[1]{\kw{InfAcc} #1}

\newcommand{\infacc}[1]{\kw{InfAcc} #1}

\newcommand{\plan}{\kw{Plan}}

\newcommand{\mt}{\kw{mt}}
\newcommand{\aplan}{\kw{PL}}
\newcommand{\outcomet}[3]{\ensuremath{\llbracket#2{\mid#1}\rrbracket_{#3}}}
\newcommand{\outcomeb}[2]{\ensuremath{\llbracket#1 \rrbracket ( #2) } }
\newcommand{\interpret}[2]{\ensuremath{\llbracket#1 \rrbracket ( #2) } }

\newcommand{\ids}{\kw{Ids}}
\newcommand{\names}{\kw{Names}}
\newcommand{\uname}{\kw{uname}}
\newcommand{\addr}{\kw{addr}}
\newcommand{\uid}{\kw{uid}}

\renewcommand{\phi}{\varphi}

\newcommand{\myparagraph}[1]{\subparagraph{#1.}}

\newcommand{\myeat}[1]{}

\newcommand{\aninst}{\kw{I}}
\newcommand{\inst}{\aninst}

\newcommand{\phone}{\kw{phone}}

\title{\vspace{26pt}Interpolation and Query Rewriting}
\author{Michael Benedikt}{University of Oxford, UK}{michael.benedikt@cs.ox.ac.uk}{}{}{}
\begin{document}
\maketitle

\vspace{-4pt}
\begin{abstract}
We overview applications of Craig interpolation and Beth definability to simplifying logical expressions or database queries. From the perspective of the theory of interpolation and definability the results give a number of new angles. First, they give a different take  on what it means to make  definability or interpolation results effective, looking at algorithms that take a proof as input and return an interpolant or explicit definition as output. Secondly, they relate interpolation and definability to preservation theorems in model theory: interpolation and definability theorems are the basis for many ``semantics-to-syntax''  results, relating a semantic property of a formula to its equivalence with a certain syntactic form. Thirdly, they motivate new forms of interpolation and definability, focusing on syntactic forms that are of interest in databases.
\end{abstract}
\vspace{-3pt}
\tableofcontents

\section{Overview} \label{sec:textoverview}
There are a number of applications of Craig interpolation in data management, but this chapter will focus on one type of application, to \emph{rewriting queries}.
The idea  is that we want to  translate a \emph{declarative source query} into
a declarative query or procedural plan in a \emph{target language} that abides by certain \emph{interface restrictions}.
By a declarative source query we will usually mean a request for information from a database, where the request is specified
using a formula of first-order logic, or its equivalent in the database query language SQL.
We  often focus on  source queries given in the language of \emph{conjunctive queries}, which corresponds
to a very simple subset of SQL.
 In logical terms these correspond to  formulas built up using only conjunction and existential quantification.
By a plan in a target language we may mean either another declarative query, or something more operational,
such as a relational algebra expression, possibly using some restricted set of physical operators.

Translating from a declarative source query to a target in this broad
sense captures a huge number of activities that take place within data management.
One of the basic features of modern database systems is to decouple
the vocabulary employed by users to   ask questions about a dataset from
the data structures that are used to implement data access:
the ``logical model'' and the ``physical model''. In  relational
databases, the logical model might be a set of tables given an SQL schema: users
interact with the data by sending SQL queries that mention these tables. The database manager
performs a translation of the queries written over this high-level vocabulary
into a program that interacts with the data using some set of data access functions.

In this ``classical query evaluation'' scenario, the relationship of the source
vocabulary to the target vocabulary is very simple: every source relation corresponds
to one or more ``physical datasource'' relations or functions.
A more complicated example comes from data integration systems. The goal there is to mask a diverse set of  datasources
by providing a single unified schema. Suppose we have a multitude of ``local databases'' with information
about different hotel chains.  A data integration system
would present to users a ``global schema'' --- an easier-to-understand  vocabulary that can model information in all of the sources.
The  system accepts queries written over the global schema and translates them into queries over the
various local sources. The ``global database'' is  virtual, implicitly defined by its relationship with  local data.

Having talked generally about what we mean by sources and targets, we mention several flavors of  languages that represent targets for a rewriting algorithm.
The most basic kind of restriction we look at is a \emph{vocabulary-based restriction}.
There the source is a query represented as a logical formula, and the target is also a logical formula,
but we limit the target's vocabulary.
We begin with a logic-based query $Q$ written using some relations $R_1 \ldots R_j$,
and want to convert it to another logic-based query
$Q_{\vec V}$ making use of a different set of relations $V_1 \ldots V_k$.
If the queries $Q$ and $Q_{\vec V}$ mention different relations,  we cannot expect
$Q$ to give the same answers as  $Q_{\vec V}$ on arbitrary instances.
But our schema will come with \emph{integrity constraints} --- usually given as a set of first-order sentences --- which restrict the possible instances of interest.
We will thus be considering equivalence only on instances satisfying the constraints.

An example of vocabulary-based restriction comes from \emph{reformulating queries over views}.
We have a collection of view relations $V_1 \ldots V_k$,
and each $V_i$ is associated with a query $Q_i$ defined using some other set of relations $R_1 \ldots R_j$.
Given a query $Q$ defined over $R_1 \ldots R_j$, the goal is to find an
equivalent query $Q_{\vec V}$ that mentions only $V_1 \ldots V_k$.
Generally, additional restrictions will be put on $Q_{\vec V}$.
For example, it should be a conjunctive query, or a union of conjunctive queries, or a relational algebra query.
The view-based query reformulation problem can thus be seen as a special case of vocabulary-based restriction, where the integrity constraints are of the form:
\[
\forall x_1 \ldots \forall x_n ~ V_i(x_1 \ldots x_n) \leftrightarrow Q_i(x_1 \ldots x_n)
\]
for $1 \leq i \leq k$.
Here the constraint with $\leftrightarrow$ could be rewritten to be a conjunction of two constraints with $\rightarrow$, one in each direction, and we will use this version below.

\begin{myexmp} \label{ex:views}
A university database has a table $\professor$ containing ids and  last names of professors, along with the name of the professor's department.
It also has a table $\student$ listing the id and last name of each student,
along with their advisor's id.

The database  does not allow users to access the $\professor$ and $\student$ tables directly,
but instead exposes a view $\professor'$ where the id attribute is dropped,
and a table $\student'$ where the advisor's id is replaced with the advisor's last name.

That is, $\professor'$ is a view defined by the formula:
\begin{align*}
\{~\lname,~\dname~\mid \exists~\profid~\professor(\profid,~\lname,~\dname)\}
\end{align*}
or equivalently by the constraints:
\begin{align*}
\forall \profid ~ \forall \lname ~ \forall \dname ~
\professor(\profid, \lname, \dname)
\rightarrow \professor'(\lname, \dname) \\
\forall \lname ~ \forall \dname ~
\professor'(\lname, \dname) \rightarrow
\exists \profid ~ \professor(\profid, \lname, \dname)
\end{align*}
$\student'$ is a view defined by the formula:
\begin{align*}
\{ \studentid, \lname, \profname \mid \\
\exists \profid ~ \exists \dname ~ \student(\studentid, \lname, \profid)
\wedge \professor(\profid, \profname, \dname) \}
\end{align*}
or equivalently by constraints:
\begin{align*}
\forall \studentid ~ \forall \lname ~ \forall \dname ~ \forall \profid ~ \forall \profname ~
[\professor(\profid, \profname, \dname)\; \wedge \\
\student(\studentid, \lname, \profid) \rightarrow \student'(\studentid, \lname, \profname) ]\\
\\
\forall \studentid ~ \forall \lname ~ \forall \profname
~ [\student'(\studid, \lname, \profname) \rightarrow \\
\exists \profid ~ \exists \dname ~
\professor(\profid, \profname, \dname) \wedge  \student(\studentid, \lname, \profid)]
\end{align*}

Consider the query  asking for the names of the advisors of a given student.
We can answer this by simply using the $\student'$ view, returning the $\profname$ attribute of a tuple in the view.
A formula over relational expression that uses $\student'$ to answer the query is called
a \emph{reformulation of the query over the views}. More generally it is
an example of a \emph{reformulation of the query over the subvocabulary $\student'$, relatively to the
background theory consisting of constraints relating $\student'$ to the $\student$}.

But a query asking for the last names of all  students that have an advisor in the history department
can \emph{not} be answered using the information in the views:
knowing the advisor's name is not enough to identify the department.

Integrity constraints need not be restricted to expressing  view definitions. A natural use of constraints is
to represent \emph{relationships between sources}, such as overlap in the data.
This overlap can be exploited to take a query  specified over a source
that \textit{a priori} does not have sufficient data, and reformulate it  over a source that
provides the necessary data.
\end{myexmp}

\myparagraph{Access Methods} A finer notion of interface than vocabulary-based restrictions is based on \emph{access methods},
which state that a relation can only be accessed by lookups where certain arguments must be provided.
One example of a restricted interface based on access methods comes from web forms.
Thinking of the form as exposing a virtual table, the mandatory fields must be filled in by the user,
while submitting the form returns all tuples that match the entered values.
Other examples include web services and legacy databases.

\begin{myexmp} \label{exone}
Consider a $\profinfo$ table available via a web form, containing information about faculty,
including their last names, office number, and employee id, but with a restricted interface that
requires giving an employee id as an input.
The interface could be implemented as a web form that requires
entering an employee's id and then pressing a submit button to get the matching records.
The query $Q$ asking for ids of faculty named $\smith$ cannot be ``completely answered'' over this
schema:
that is, there is no function over the available data in this schema which is equivalent to $Q$.

But suppose another web form provides access to a $\univdirectory$ table containing
the employee id and last name of every university employee,
with an interface that allows one to access the entire contents of the table.
Then we can reason that $Q$ is answerable using the information in this schema:
a plan would pull tuples from the $\univdirectory$ form and check them within the $\profinfo$ form
to see if they correspond to a faculty member.
\end{myexmp}

In Example~\ref{exone}, reasoning about access considerations was straightforward,
but in the presence of more complex schemas, we may have to chain several inferences,
resulting in a plan that makes use of several auxiliary accesses.

\begin{myexmp} \label{extwo}
Suppose we have two directory data sources with overlapping information.
One source exposes information from $\udirectory_1 (\uname, \addr, \uid)$ via an
access method requiring a $\uname$ and $\uid$.
There is also a ``web table'' $\ids(\uid)$ with no access restriction, that makes available the set
of $\uid$s
(hence we have a ``referential integrity constraint'' saying that every $\uid$ in $\udirectory_1$ matches a $\uid$ in $\ids$).
The other source exposes $\udirectory_2(\uname, \addr, \phone)$, requiring a $\uname$ and $\addr$,
and also a web table $\names(\uname)$ with no access restriction that reveals all $\uname$s in $\udirectory_2$ 
(that is, a constraint that each $\uname$ in $\udirectory_2$  appears in $\names$).
There is also a constraint saying that
each $\uname$ and $\addr$ in  $\udirectory_2$ appears in $\udirectory_1$.

A query asking for all phone numbers in the second directory could be written:
\[
Q=\{ \phone \mid \exists  \uname ~  \exists \addr ~ \udirectory_2(\uname, \addr, \phone) \}.
\]

There is a plan that implements this query:
it gets all the $\uid$s from $\ids$ and  $\uname$s from $\names$ first,
puts them into the access on $\udirectory_1$,
then uses the $\uname$ and $\addr$ of the resulting tuples to get the phone numbers in $\udirectory_2$.
\end{myexmp}

We emphasize that the notion of rewriting we consider here
is getting target queries  which give the same answers as the original query.
This  means that if we have a query asking for the office number
of all professors with last name $\smith$, the plan produced should
return all tuples in the answer, even if access to the $\professor$ relation is limited.
This contrasts with a line of work in data integration that considers
the broader question of how to answer any query ``as much as possible given the available data'':
 how to get the \emph{certain answers}
or how to compute the \emph{maximally contained query}~\cite{dataint}.

\myparagraph{Rewriting via Interpolation: the Meta-Algorithm}
We will now
explain the connection of rewriting to interpolation and Beth definability.
There is a general template for solving a rewriting problem:

\begin{enumerate}
\item  Isolate a \emph{semantic property}  that any input query $Q$ must have with respect to the target $\T$
and constraints $\Sigma$ in order to have a reformulation of the desired type.
\item Express this property as a \emph{proof goal} (in the language we use later on, an \emph{entailment}):
        a statement that a certain formula follows from another formula.
\item Search for a proof of the entailment within a given proof system.
       For example, one could use tableau proofs or resolution proofs, both well-known proof systems within computational logic.
\item From the proof, extract an \emph{interpolant} using an interpolation algorithm.
        We will review some standard interpolation algorithms and also  present new ones.
\item Convert the interpolant to a reformulation. 
\end{enumerate}

As we will see below, the properties of an interpolant will play a role in the ability to convert to a reformulation, and different notions of interpolation, which control how similar an interpolant is to the statements we interpolate over, will be needed for different variations of reformulation.
This approach is very general and can be applied to a variety of proof systems and restrictions on the target, 
with different target languages corresponding to different entailments.
We prove theorems saying that the method is complete:
there is a reformulation exactly when there is a proof of the semantic property.
These completeness theorems give as a consequence a \emph{definability or preservation theorem}:
a query $Q$ has a certain kind of reformulation if and only if it has a certain semantic  property.
Such theorems are well-known in model theory, and indeed our theorems can be seen
as ``database versions'' of the preservation theorems that are known from classical model theory
textbooks, \eg~\cite{ChangKeisler}.

The fact that interpolation theorems can be used to prove preservation theorems is
certainly not new. As mentioned earlier in this volume,  William Craig used
interpolation to prove Beth's definability theorem, and all of our results can be seen
as generalizations of Craig's technique: see the bibliographic remarks at the end of this chapter for further details.

What we emphasize in this chapter, as in \cite{interpbook} is that the connection between reformulation, interpolation,
and preservation can be made effective: \emph{interpolation algorithms yield reformulation algorithms}. 
Actually, the interpolation technique yields many of the algorithmic results about reformulation that had appeared in the  database setting previously.

\section{Preliminaries} \label{sec:prelims}

We start with some preliminaries about relational databases and logic. A good reference for this is \cite{AHV}.

\myparagraph{Structures, Instances, and Range of Quantification}
A  \emph{schema} is a finite set of relation symbols, each associated to a number (the arity), along with a set of constant symbols.
A database \emph{instance}  $\aninst$ for a schema $\aschema$ assigns to every relation  $R$ in $\aschema$
of arity $n$ a collection of $n$-tuples  $\interpret{R}{\aninst}$.

We call $\interpret{R}{I}$ the \emph{interpretation} of $R$ in $\aninst$.
An association of a database relation  $R$ with a tuple $\vec c$ of the proper arity will be referred to as a \emph{fact}.
A database instance can equivalently be seen as a collection of facts.
An instance of a schema that has only a single relation $R$ is a \emph{relation instance}.
The \emph{active domain} of an instance $\aninst$, denoted $\adom(\aninst)$, is the union of the one-dimensional projections of all interpretations of relations:
that is, all the elements that participate in some fact of $\aninst$.

We can extend the notion of a schema to allow constant symbols: now an instance also maps each constant symbol to a value.
We can also extend the notion of a schema to allow \emph{integrity constraints}: in general these are an arbitrary subset of the instances. An instance of a schema with constraints is one where  all integrity constraints of the schema are satisfied.

\begin{myexmp} 
\label{ex:relationalschema}
Suppose our schema consists of only one unary relation $\uemployee$, containing the ids of university employees.  
One possible instance $\aninst$ interprets $\uemployee$ by the singleton set $\{e_0\}$.
We can alternatively define $\aninst$ by the set of facts $\{\uemployee(e_0)\}$.
In this case we have $\adom(\aninst)=\{e_0\}$.
\end{myexmp}

A \emph{query} (over a given vocabulary) and output arity $m$ is a function from instances of the schema to a set of $m$-tuples. A \emph{Boolean query} is the case where $m=0$: thus the output is a set of $0$-tuples. There are only two $0$-tuples: the empty set, which we identify with False, and the singleton $0$-tuple, which we identify with True.
If $m>0$ we talk of a \emph{non-Boolean query}.

Classical logic considers \emph{structures} rather than \emph{instances}.
A structure consists of a set, the \emph{domain} of the structure,
interpretations for each relation as sets of tuples with values in the domain, and an interpretation for each constant as a single
element of the domain.
 
\begin{myexmp}
\label{ex:structure}
Returning to the schema in Example~\ref{ex:relationalschema}, one structure $M$ that conforms to it
consists of   a two-element domain $\{e_0, f_0\}$, with $\uemployee$
interpreted as $\{e_0\}$.
\end{myexmp}

In the \emph{classical semantics} of first-order logic, quantifiers range over the domain of the structure.
In databases, quantifiers are usually given the \emph{active domain} semantics, in which the quantified variable
  ranges over the union of the values in the interpretations of relations.
The active domain semantics can be used
to give a meaning to a sentence in an instance, since  the meaning
only depends on the instance, not some 
domain in which the instance sits. 

A \emph{relational atom} is an expression $R(u_1 \ldots u_n)$ where $R$ is an $n$-ary relation symbol and the $u_i$ are either constants or variables. 
We will use  \emph{first-order logic with relativized quantifiers}, $\rqfo$.
$\rqfo$ is built up from equality and relational atoms
via the Boolean operators and the quantifications:
\[
\exists \vec x ~ R(\vec s, \vec x) \wedge \phi(\vec s, \vec x, \vec t)
\]
 and
\[
\forall \vec{x} ~ R(\vec s, \vec x) \rightarrow \phi(\vec s, \vec x, \vec t)
\]
for $R$ a relational symbol and $\phi$ an $\rqfo$ formula.
Those familiar with active domain semantics will be able to see that it is subsumed by $\rqfo$,
in the sense that an active domain semantics formula can be converted to an equivalent $\rqfo$ formula under classical semantics:
the active domain interpretation of the quantification $\exists x ~ \phi(x, \vec y)$ is a disjunction $\bigvee_i \exists \vec u_i ~ R_i(\vec u_i) \wedge \phi(x, \vec y)$, where $x$ is one of the variables in each $\vec u_i$.
Similarly a quantification $\forall x ~ \phi$ under active domain
semantics translates into a conjunction of  relativized universal
quantifications.

We can similarly talk about equality-free $\rqfo$ formulas, by disallowing equality atoms.
By convention, we allow the prefix of existential quantifiers
$\exists x_1 ~ \ldots ~ \exists x_n$ to be empty, and similarly for universal quantifiers.
In this way we can capture an atom $R(x_1 \ldots x_n)$ and negated atoms
$\neg R(x_1 \ldots x_n)$ as specialized $\rqfo$ formulas
$R(\vec x) \wedge \true$ and $R(\vec x) \rightarrow \false$, respectively.
Thus for equality-free formulas we can in fact assume that the only base formulas are $\true$ and $\false$.

The truth value of an $\rqfo$ sentence without constants is well-defined over an instance, since
 from the instance we can determine the range of each relativized quantifier.
Similarly the truth value for an $\rqfo$ formula without constants is well-defined given an instance
and a mapping of the free variables to some values (a \emph{variable binding} or just \emph{binding} for short).

An $\rqfo$ formula is in \emph{Negation Normal Form} (NNF) if negation is only applied to relational atoms. By simple rules, one can convert any $\rqfo$ formula into NNF.

\myparagraph{Relational Algebra}
$\rqfo$ gives a natural way of defining Boolean queries over finite instances. For non-Boolean queries we review another standard formalism.
Given a relational vocabulary, we can build up \emph{relational algebra} a language for defining queries. A relational algebra expression is associated with an output type, which is a finite set of attribute names. The atomic expressions are the relation symbols for names in the vocabulary, and for a relation symbol $R$ of arity $n$, the output type has names $\{1 \ldots n\}$.
The language is closed under the operations:
\begin{compactitem}
    \item 
difference, intersection, and union: if $E$ and $F$ are expressions having the same output $T$, $E-F$, $E \cup F$ and $E \cap F$ are new expressions also having type $T$
\item projection: if $E$ is an expression with output type $a_1 \ldots a_n$ and $S$ is a subset of $a_1 \ldots a_n$ then $\pi_S(E)$ is an expression with output type $S$
\item renaming: if $E$ is an expression with output type $S$ and $F$ is a bijection from
$S$ to $S'$ then $\rho_{F}(E)$ is an expression with output type $S'$
\item product: if $E$ and $F$ are expressions with disjoint output types $S$ and $S'$, then $E \times F$ is an expression with output type $S \cup S'$
\item selection: if $E$ is an expression with output type $\{a_1 \ldots a_k\}$ then $\sigma_{a_i=a_j} E$ is an expression with the same output type.
\end{compactitem}

The semantic function for relational algebra maps each term $E$ with output type $S$ to a query whose input is an instance over the input vocabulary and whose output is an instance of the signature consisting of an $S$-ary relation symbol, along with a mapping of positions to elements of $S$.
For atomic terms it is the identity. The inductive definition of each operator is also straightforward. 

We note a superficial difference between relational algebra and logic-based formalisms like $\rqfo$. Relational algebra uses \emph{named notation}: it deals with instances where in each tuple of a relation, the different components are referred to via attribute names. As opposed to logical formalisms, where we have \emph{positional notation}: each relation has some arity $m$ and different components are referred to by a position $1 \leq i \leq m$. It is easy to map between named and positional notation. Below we will abuse notation slightly and talk about a formalism in one notation being equivalent in expressiveness to a formalism in the other: this is up to some canonical way of mapping between names and positions. One can find the formalization of such mappings in \cite{AHV}.

A term with empty output type returns either the empty instance or the instance consisting of a single empty tuple: it can thus be considered as a Boolean query.

\begin{proposition} \label{prop:relalgandrqfo} A relational algebra expression of Boolean type can be mapped to an equivalent $\rqfo$ expression, and vice versa.
\end{proposition}

\begin{proposition} \label{prop:relalgsafe} The output of a relational algebra expression on a finite instance is a finite relation.
\end{proposition}

A query is \emph{safe} if on every finite instance, there are only finitely many variable bindings that satisfy it.

There is a converse to Proposition \ref{prop:relalgsafe}: for any $\rqfo$ formula, if the query it defines is safe,   then there is a relational algebra expression that represents the same query.

A relational algebra expression that does not use the difference operator is said to be a $\uspj$ expression (for union-selection-projection-join). If further it does not use union, it is an $\spj$ expression.

\section{First Example: Query Reformulation over a Subvocabulary}
We now drill down into the first example of reformulation-via-interpolation: vocabulary-based reformulation.

\subsection{Reformulating Queries over Restricted Vocabularies: Definitions}

Let $\S$ be a collection of relations, $\Sigma$ a set of integrity constraints, and $\T$ a subset of $\S$.  
Given a query $Q$ specified  by a logical formula with $n$ free variables over $\S$, 
our first goal will be to get a \emph{relativized-quantifier first-order reformulation of $Q$ over $\T$ with respect to $\Sigma$}. 
This means a query $Q_\T$ given by a relativized-quantifier first-order formula using only the relations in $\T$ 
such that for every instance $I$ satisfying $\Sigma$, 
\[
\forall x_1 \ldots x_n ~ Q(x_1 \ldots x_n) 
\leftrightarrow Q_\T(x_1 \ldots x_n)
\]
 holds in $I$.
We use the same notation when $Q$ is a relational algebra query, 
identifying its output attributes with free variables.
Thus a reformulation is another query $Q_\T$ that is \emph{equivalent to $Q$ $\wrt$ $\Sigma$} or \emph{answers $Q$ $\wrt$ $\Sigma$}, 
meaning that they have the same output on instances satisfying $\Sigma$.

\subsection{From a Semantic Property to a First-Order Reformulation}
Let us return to the ``meta-algorithm'' for reformulating a query $Q$ with respect to a target language and
a set of constraints, mentioned in the introduction. Recall that step (1)
is to identify a property that $Q$ must have in order to admit the desired reformulation. 
Clearly, for $Q$ to have any reformulation over $\T$, its output should depend only on the interpretations
of the relations in $\T$. This notion of ``implicit definability'' has been formalized by Segoufin and Vianu
in \cite{SVconf}, as the notion of 
\emph{determinacy} for database views and queries. We extend
the definition here to the setting where there are constraints given by $\rqfo$ sentences.

If $\Sigma$ is a collection of $\rqfo$ constraints, we say that an $\rqfo$
query $Q$ over $\S$ is \emph{determined over $\T$ relative to $\Sigma$} if:

\begin{quote}
For any two instances $I$ and $I'$ that satisfy $\Sigma$ 
and have the same interpretation of all relations in $\T$ (that is, they have the same $T$-facts for each $T \in \T$)
$\interpret{Q}{I}=\interpret{Q}{I'}$.
\end{quote}

\begin{myexmp} \label{ex:det}
We consider the case where  the vocabulary consists of
the  $\department$ and $\employee$ relations, the constraints 
$\Sigma$ are  inclusion dependencies from $\employee$ to $\department$ (on $\did$)
and from $\department$ to $\employee$ (on $\mgrid$ and $\did$), while the
target signature $\T$ consists of only  $\department$.

Suppose our query $Q$ asks for the ids of all employees. We claim that $Q$ is \emph{not} determined over $\T$ relative to $\Sigma$.
To see this, consider two instances $I$ and $I'$, such that $I$ consists of facts:
\begin{align*}
\{ \employee(123, \jones , 1112), \employee(134, \smith, 1112), \\
\department(1112, \mathematics, 134)\}
\end{align*}
and $I'$ consists of facts  
\begin{align*}
\{ \employee(134, \smith, 1112),\department(1112, \mathematics, 134)\}
\end{align*}
$I$ and $I'$ both  satisfy
the schema constraints, and they have the same restriction to $\department$, but the output of
$Q$ on $I$ includes employee id $123$, while $Q$ evaluated on $I'$ does not.

We have shown via witnesses $I$ and $I'$  that  $Q$ is not
determined over $\T$ $\wrt$ $\Sigma$.
\end{myexmp}

\noindent A special case of interest is where the relations in the restricted
vocabulary $\T$ are ``virtual tables'' defined by \emph{view definitions}. 
For each $V \in \T$ of arity $n$ there is a corresponding query  $Q_V(x_1 \ldots x_n)$, 
while the constraints $\Sigma$ consist only of the statements 
that a view contains exactly those tuples that satisfy its definition:
$\forall \vec x ~ V(\vec x) \leftrightarrow Q_V(\vec x)$.
In this case, determinacy of another query $Q$ with respect to $\T$ and $\Sigma$
can be rephrased as:

\begin{quote}
For any two instances $I$ and $I'$, if $I$ and $I'$ yield the same results for each $Q_V$,
then they yield the same results for $Q$.
\end{quote}
Segoufin and Vianu use the terminology ``$Q$ is determined by the views $\{ Q_V : V \in \T\}$''. 
We will apply our methodology to the semantic property of determinacy,
turning it into a proof goal, and showing that reformulations can be read off from these proofs witnessing the proof goal.

\myparagraph{Translating the Semantic Property to an Entailment}
Returning to Step (2) of our ``meta-algorithm'' from the introduction, 
we write out the semantic property (in this case, determinacy) as a proof goal. 
We now formalize what this means.

Recall  that an \emph{entailment} is a statement of the form $\lambda(\vec x) \models \rho(\vec x)$, 
which means that for every instance $I$ and any binding $\abind$ of $\vec x$, 
if instance $I$ with binding $\abind$ satisfies $\lambda$ then it satisfies $\rho$.
By convention, we write $\models \rho$ to mean $\true \models \rho$:
that is, $\rho$ holds in every structure.
Let us extend our original  signature for the constraints $\Sigma$ and 
the query $Q$ by making a copy $R'$ of every relation $R\in \S$. 
Let $Q'$ be the copy of $Q$ on the new relations, and $\Sigma'$ be the copy
of the constraints $\Sigma$ on the new relations.
Our assumption of determinacy of $Q$ can be restated as an entailment:
\begin{align*}
\models \forall \vec x ~ [\Sigma \wedge  \Sigma' \wedge  (\bigwedge_{T \in \T} \forall \vec y  ~ T(\vec y) \leftrightarrow
T'(\vec y) ) \wedge Q(\vec x) \rightarrow Q'(\vec x)]
\end{align*}
Or, rewriting,
\begin{align*}
\Sigma \wedge Q(\vec x) \models 
[(~ \bigwedge_{T \in \T} \forall \vec y  ~ T(\vec y) \leftrightarrow
T'(\vec y) ~ ) \wedge \Sigma' ] \rightarrow Q'(\vec x)
\end{align*}
This is the \emph{entailment associated with determinacy}.

\myparagraph{How Interpolation for the Entailment Gives a Reformulation}
Going back to our general plan:
\begin{enumerate}
\item  We have isolated a \emph{semantic property}  that an  input query $Q$  must have
with respect to the target $\T$
and constraints $\Sigma$ in order to have a relativized-quantifier first-order reformulation: 
it must be determined over $\T$ \wrt $\Sigma$.
\item We have expressed this property as an entailment:
\begin{align*}
\Sigma \wedge Q(\vec x) \models [(\bigwedge_{T \in \T} \forall \vec y  ~ T(\vec y) \leftrightarrow T'(\vec y) ) \wedge \Sigma']
\rightarrow Q'(\vec x)
\end{align*}
\end{enumerate}
We now  relate this  to \emph{interpolants}.

If $\lambda \models \rho$ is an entailment, 
an \emph{interpolant} is a formula $\chi$, such that:

\begin{itemize}
\item $\lambda \models \chi$ and $\chi \models \rho$ 
\item Every relation in $\chi$ occurs in both $\lambda$ and $\rho$.
\end{itemize}

\emph{Interpolation theorems} state that any entailment for formulas in a logic $L$ has an interpolant
in logic $L$, possibly satisfying additional conditions.

To instantiate the meta-algorithm for this setting, we need to:

\begin{itemize}
\item show that an interpolant for this entailment gives the desired reformulation
\item develop an algorithm to extract these interpolants
\end{itemize}

We start with the first item, the argument that an interpolant for the entailment above represents the first-order reformulation we want:
\begin{proposition}[Variation of \cite{craig57beth}] \label{prop:interpgivesreform} 
Suppose $Q$ is a first-order logic formula and  $\Sigma$ is a first-order sentence.
Let $\chi$ be any interpolant for the entailment
\begin{align*}
\Sigma \wedge Q \models
[(\bigwedge_{T \in \T} \forall \vec y  ~ T(\vec y) \leftrightarrow
T'(\vec y)) \wedge \Sigma' ] \rightarrow Q'
\end{align*}
Then $\chi$  uses only the relations in $\T$, and is equivalent to $Q$ for structures that satisfy $\Sigma$.
\end{proposition}

In particular, if $\chi$ is an interpolant in relative-quantifier first-order logic, we have
found an RQFO reformulation of $Q$ relative to $\Sigma$.

The proof of the proposition is quite straightforward:
by the definition of an interpolant, $\chi$ uses only the relations common to the left and right, which
are those in $\T$. And it is easy to see that $\Sigma$ entails that $\chi$ is equivalent to $Q$. For one direction
we use that $\Sigma \wedge Q \models \chi$, by the definition of an interpolant. For the other part of the equivalence, using the fact that $\chi$ entails the right side, and setting the primed and unprimed relations equal we see
that $\Sigma \wedge \chi \models Q$.

We have thus completed all the steps of the meta-algorithm, provided that we 
have a way to get $\rqfo$ interpolants from an entailment between $\rqfo$ formulas.

\subsection{The Interpolation Theorem We Need: Relativized-Quantifier Interpolation} \label{subsec:craigadom}
We discuss the modification of Craig interpolation needed to take
a proof witnessing an entailment for relativized-quantifier first-order formulas and
produce an $\rqfo$ interpolant.
This will allow us to instantiate the meta-algorithm for $\rqfo$  constraints:

The interpolation result we need is a relativized-quantifier version of Craig interpolation, 
a variant of a result proven originally by Martin Otto.

\begin{theorem}[\cite{otto}][Relativized-Quantifier Craig Interpolation Theorem] \label{thm:adomcraig}
If $\lambda$ and $\rho$ are $\rqfo$ formulas such that $\lambda \models \rho$, 
then there is an interpolant $\chi$ in $\rqfo$.
Furthermore, if $\lambda$ and $\rho$ do not use equality, neither does $\chi$.
\end{theorem}

In \cite{interpbook}, an effective version of this was proven, where we assume a proof of the entailment in a certain proof system. In \cite{interpbook} we use analytic tableaux as a proof system: this proof system and the interpolation algorithm for it are discussed in
\refchapter{chapter:firstorder}. But in the results we  will just refer to a \emph{suitable proof system}, leaving the details to the references.

\begin{theorem}[Effective Relativized-Quantifier Craig Interpolation Theorem] \label{thm:effadomcraig}
If $\lambda$ and $\rho$ are $\rqfo$ formulas \emph{and} we have a suitable proof of $\lambda \models \rho$, 
then \emph{in polynomial time in the proof} we can find an interpolant $\chi$ in $\rqfo$.
Furthermore, if $\lambda$ and $\rho$ do not use equality, neither does $\chi$.
\end{theorem}

\myparagraph{Putting It All Together}
We are  ready to give the  instantiation of the general methodology. Instantiating it to this setting, we see that
to find an $\rqfo$ reformulation of $Q$ with respect to subvocabulary $\T$ and constraints $\Sigma$ we should:
\begin{itemize}
	\item Search for a proof that the semantic property determinacy holds. 
	That is, find a proof witnessing
	\[\Sigma \wedge Q(\vec x) \models
	[(\bigwedge_{T \in \T} \forall \vec y  ~ T(\vec y) \leftrightarrow
	T'(\vec y) ) \wedge \Sigma' ] \rightarrow Q'(\vec x)
	\]
	\item Use the relativized-quantifier tableau-based interpolation algorithm to produce an interpolant $Q_\T(\vec x)$.
	\item As  shown in Proposition~\ref{prop:interpgivesreform}, such an interpolant will give the reformulation we want.
\end{itemize}

We also know that to have an $\rqfo$ reformulation, $Q$ must be determined  over $\T$ $\wrt$ $\Sigma$, and thus
there must be a proof of the entailment for determinacy. This means
we have proven the following \emph{effective reformulation theorem}:

\begin{theorem}\label{thm:effectivepbdadom}
To find a reformulation of an $\rqfo$ query $Q$ with respect to subvocabulary $\T$ and $\rqfo$ constraints $\Sigma$, 
it is sufficient to find a proof (in a suitable proof system) witnessing the entailment:
\[
 \Sigma \wedge  Q \models  [(\bigwedge_{T \in \T} \forall \vec y  ~ T(\vec y) \leftrightarrow
        T'(\vec y) ) \wedge \Sigma' ] 
  \rightarrow Q'
\]
where $\Sigma'$ is formed from $\Sigma$ by replacing each relation $R$ with $R'$, 
and $Q'$ is formed similarly from $Q$.

From any such proof, we can effectively produce a reformulation.
\end{theorem}

Theorem \ref{thm:effectivepbdadom} reduces searching for a reformulation to searching for a proof.
Since the existence of a reformulation implies that the entailment above holds, and the entailment above captures determinacy:
\begin{corollary}[Equivalence of Reformulations, Entailments, and Semantic Properties] \label{cor:iffpbd}
An $\rqfo$ query $Q$ has an $\rqfo$  reformulation with respect to subvocabulary
$\T$ and constraints $\Sigma$ if and only if $Q$ is determined over subvocabulary $\T$ with respect to $\Sigma$ 
if and only if the entailment in Theorem~\ref{thm:effectivepbdadom} holds.
\end{corollary}

\section{Vocabulary-Based Reformulation with Positive Existential Queries} \label{sec:vocabconjunctive}

Expanding on our general program outlined in the introduction, we show
what happens when we restrict the target language for a rewriting.
Recall  that a \emph{positive existential formula with inequalities} 
($\posexistsineq$ formula) is a formula built up using only $\exists, \wedge, \vee$ from relational atoms
and inequalities. We also consider the formula $\False$ to be positive existential with inequalities.

Given an $\rqfo$  formula $Q$, restricted vocabulary $\T$, and constraints $\Sigma$ given by $\rqfo$ sentences, 
we are interested in getting a $\posexistsineq$ reformulation of $Q$ over $\T$ with respect to $\Sigma$. 
This means we want a $\posexistsineq$ formula over $\T$ that agrees with $Q$ for instances satisfying the constraints.

\myparagraph{The Semantic Property for $\posexistsineq$  Reformulation} 
Following the ``meta-algorithm'' from the introduction, we start by finding the appropriate
semantic property that an input $Q$ must have to admit  a $\posexistsineq$  reformulation.
\cite{NSV} isolated such a property, which we call monotonic determinacy\footnote{\cite{NSV} used the term ``monotone''}.
We say that a query $Q$ is \emph{monotonically-determined over $\T$} relative to $\Sigma$ if:

\begin{quote}
for any two instances $\inst_1, \inst_2$ that satisfy $\Sigma$ and such that
for all relations $T \in \T$, $\interpret{T}{\inst_1} \subseteq \interpret{T}{\inst_2}$, then $\interpret{Q}{\inst_1} \subseteq \interpret{Q}{\inst_2}$.
\end{quote}

\myparagraph{The Entailment Corresponding to the Semantic Property} 
Proceeding to the second step of the meta-algorithm, we will express this semantic property as an entailment.
Let $\Sigma'$ be a copy of the constraints $\Sigma$ where each occurrence of a relation $R$ in $\S$ has been replaced by a copy $R'$.

Monotonic determinacy of a first-order query $Q$ over $\T$ relative to $\Sigma$ can be restated as saying that the following
sentence holds on all instances:
\begin{align*}
\forall \vec x ~ [\Sigma \wedge  \Sigma' \wedge (\bigwedge_{T \in \T} \forall \vec y ~ T(\vec y) \rightarrow T'(\vec y))
\wedge Q(\vec  x)] \rightarrow Q'(\vec x)
\end{align*}

Above, $Q'$ and $\Sigma'$ are again the result of changing unprimed relations $R$ to their primed counterparts $R'$ within $Q$ and within $\Sigma$ respectively.
Observe that if a $\posexistsineq$  formula  over $\T$ is true on an instance $\inst$, then it is true on any
instance $\inst'$ which only adds tuples to the relations in $\T$.
It is thus easy to see that a sufficient condition for monotonic determinacy of $Q$ over $\T$ \wrt $\Sigma$
is that  $\Sigma$ implies the sentence $\forall \vec x ~ Q(\vec x) \leftrightarrow \phi(\vec x)$, where
$\phi$ is a $\posexistsineq$ formula mentioning only relations in $\T$.

We highlight the difference from the entailment for determinacy: here we only have implication in the ``forward'' direction,
from unprimed to primed, while for determinacy we have implications in both directions.

\myparagraph{Generating $\posexistsineq$  Reformulations from Proofs of the Entailment}
We will show that by applying interpolation to proofs of the entailment for monotonic determinacy, 
we get $\posexistsineq$ reformulations.
In doing so, we will prove the following analog of Theorem~\ref{thm:adomcraig}:

\begin{theorem}\label{thm:plp}
If the constraints $\Sigma$ are in $\rqfo$  and $\rqfo$ query $Q$ is monotonically-determined over $\T$, 
then there is a $\posexistsineq$ formula $\phi(\vec x)$ using only relations in $\T$ such
that $\Sigma \models \forall \vec x ~ [ Q(\vec x) \leftrightarrow \phi(\vec x)]$.
Furthermore, if the constraints and the query $Q$ do not make use of equality, 
then $\chi$ can be taken to be positive existential (without inequalities).
\end{theorem}

We refer to this as the \emph{Projective Monotone Preservation Theorem}. 
 The adjective ``Projective'' emphasizes that we are dealing with a subset of the signature, 
as opposed to many preservation theorems one encounters in logic textbooks and papers, 
which deal with syntactically characterizing a semantic property involving the entire  signature.

The steps in proving Theorem~\ref{thm:plp} will follow the meta-algorithm.
We will prove  a modification of the Craig Interpolation Theorem, the Relativized-Quantifier Lyndon Interpolation Theorem, 
and our proof will provide an algorithm for generating interpolants of a special form from an entailment.
We then show that the interpolants produced by this algorithm, 
when applied to a re-arrangement of the entailment corresponding to the semantic property of monotonic determinacy,
will give us the $\posexistsineq$  reformulation of our query $Q$ with respect to the constraints.

\myparagraph{The Interpolants Required for $\posexistsineq$  Reformulation}
Before giving the argument, we provide some motivation.
We can restate the entailment for monotonic determinacy as:
\[
 \Sigma  \wedge (\bigwedge_{T \in \T} \forall \vec y ~ T(\vec y) \rightarrow T'(\vec y) ) \wedge Q(\vec x)
\models (\Sigma' \rightarrow Q'(\vec x))
\]
As before, $\Sigma'$ is a copy of the constraints on primed versions of the relations, 
and $Q'$ is a copy of the query $Q$ on primed versions of the relations.
The common vocabulary on the left and right consists of exactly the relations $T'$ for $T \in \T$.

Further, we note that these common relations only occur on the right of an implication on the left-hand side.
Writing out the implication $A \rightarrow B$ as $\neg A \vee B$, we see that the common relations would not occur within a negation on the left side.
Informally, we can say that these relations ``occur positively'' on the left hand side in the original formula.
Thus we want to show that the interpolant will also contain these common relations positively.
To do this, we need a formal definition of ``occurring positively'' that applies to arbitrary  first-order formulas,
and a version of interpolation that connects the relations occurring positively in the interpolant with those that occur
positively in both sides of the entailment.

Formally, we say that a relation \emph{occurs positively} in a
first-order formula if it occurs in the scope of an even number of negations, 
when we rewrite the formula to use only the quantifiers and the connectives $\wedge, \vee, \neg$.
A relation occurs  negatively in a formula if it occurs in the scope of an odd number of negations.
Intuitively if a relation occurs only positively, then  the set of solutions to the formula can only increase or stay the same
as tuples are added to that relation (this is easy to check if the relation does not appear under any negations at all).  

The interpolation theorem  that tracks which occurrences are positive is the following result, 
which is a strengthening of the Relativized-quantifier Craig Interpolation Theorem, 
Theorem~ \ref{thm:adomcraig}:

\begin{theorem}[Relativized-quantifier Lyndon Interpolation Theorem]
 \label{thm:lit}
Suppose $\lambda$ and $\rho$ are in $\rqfo$ and $\lambda \models \rho$.
Then there is an $\rqfo$ interpolant $\chi$ for the entailment
such that  a relation occurs positively in $\chi$ only if it occurs positively in both $\lambda$ and $\rho$,
and a relation occurs negatively in $\chi$ only if it occurs negatively in both $\lambda$ and $\rho$.
Furthermore
if equality does not occur in $\lambda$ or $\rho$, then it does not occur in $\chi$, so in particular
$\chi$ can not contain inequalities.
\end{theorem}

Again, there is a well-known version of this for the classical semantics of first-order logic, Lyndon's
interpolation theorem~\cite{lyndon59}.

The construction that witnesses the Relativized-quantifier Lyndon Interpolation Theorem, 
Theorem~\ref{thm:lit}, is a variant of the one for the Relativized-quantifier CIT, which is discussed \refchapter{chapter:firstorder}.

\myparagraph{Getting $\posexistsineq$ and Positive Existential Reformulations via Interpolants}
Theorem~\ref{thm:lit} shows that we can extract a certain kind of interpolant from a tableau proof 
witnessing an entailment corresponding to monotonic determinacy.
We are now ready to instantiate the last step of our meta-algorithm,
extracting a reformulation from an interpolant.
This will complete the proof of Theorem~\ref{thm:plp}.

 Apply Theorem~\ref{thm:lit} to the entailment
\[
 \Sigma  \wedge (\bigwedge_{T \in \T} \forall \vec y ~ T(\vec y) \rightarrow T'(\vec y) ) \wedge Q(\vec x)
\models (\Sigma' \rightarrow Q'(\vec x))
\]
 We can conclude that there is an $\rqfo$ interpolant $\gamma$ mentioning
only relations in the primed copy of $\T$, where these relations only occur positively. We can assume
$\gamma$ is built up from $\True$ and $\False$ via connectives and relativized quantifiers $\forall \vec x  ~ R(\vec x) \rightarrow \phi$
and $\exists \vec x ~ R(\vec x) \wedge \phi$.
We claim that any $\rqfo$ formula in which all  relations occur positively must be equivalent to a positive existential formula.
To see this,  convert an $\rqfo$ formula to NNF.
If the resulting formula $\gamma$ had any relativized universal quantifier, then consider an outermost quantification of
the form $\forall \vec x ~ (R(\vec x) \rightarrow \phi)$.
$R$ must occur negatively in this formula.
But then it must occur negatively within $\gamma$, since existential quantification and the positive Boolean operators preserve the polarity of a subformula,
and this contradicts the assumption on $\gamma$.
Hence the Negation Normal Form of $\gamma$ can not contain any relativized universal quantifier, and thus must be $\posexistsineq$.

We say $\posexistsineq$ above, rather than positive existential, because the argument applies
only to relations $R$ of the schema, not equality. 
If the equality symbol does not appear in the entailment we know that it is not generated in the interpolant.
Therefore when equality does not occur in the constraints or the query, 
we can strengthen the conclusion to be that the interpolant is positive existential.

This completes the proof of the Projective Monotone Preservation Theorem, Theorem~\ref{thm:plp}.

\myparagraph{Application to View-Based Query Reformulation}
Let us look at the setting where there are no constraints other than those that come from views defined by conjunctive queries.  
As a  corollary of  Theorem~\ref{thm:plp} we have:

\begin{corollary}
Suppose $Q$ is a CQ and $V_1 \ldots V_n$ are views defined by arbitrary equality-free $\rqfo$ formulas.
Then, $Q$ is monotonically-determined in  $V_1 \ldots V_n$ 
if and only if there is a positive existential reformulation of $Q$ in terms of  $V_1 \ldots V_n$.
\end{corollary}

\section{Vocabulary-Based Existential Reformulation}  \label{sec:vocabexistential}
At this point, we have proven a statement about first-order reformulations and one about  positive existential reformulations.
What about queries that can be reformulated using \emph{existential formulas}? 
That is, formulas that are built up from atoms and negated atoms by positive Boolean operators and existential quantification.
There are conjunctive queries that are equivalent to existential formulas but not to positive existential ones. 
For example, in the absence of any constraints
$\exists x ~ S(x) \wedge \neg R(x)$ is not equivalent to a positive existential formula. 
Can we use a similar technique to detect which formulas are equivalent to an existential formula with
respect to a set of constraints, and if so find such a reformulation? 
We give a positive answer to this below, 
restricting for simplicity to equality-free $\rqfo$.

\myparagraph{The Semantic Property for Existential Reformulation}
As before, let $\Sigma$ be a set of integrity constraints  in equality-free $\rqfo$, 
and $\T$ a subset of the relations of $\aschema$.
We start by isolating a property that $Q$ must have in order to possess an existential reformulation.

We say that a query $Q$  is \emph{\IS-monotonically-determined over $\T$} relative to $\Sigma$ if:

\medskip 

Whenever we have two instances $\inst_1, \inst_2$ that satisfy $\Sigma$ and such that $\inst_1$ is an induced
subinstance of $\inst_2$
then $\interpret{Q}{\inst_1} \subseteq \interpret{Q}{\inst_2}$.

\medskip 
An instance $\inst_1$ is an induced subinstance of $\inst_2$ means that
$\inst_2$ contains all facts of $\inst_1$ and $\inst_2$ does not add any facts over the active domain of $\inst_1$.

Note that if an existential formula over $\T$ is true on an instance $\inst$, 
then it is true on any instance $\inst'$ which only adds tuples to the relations in $\T$ 
and never ``destroys a negated assertion about a relation of $\T$ holding in $\inst$''.
From this we see that if a formula is equivalent to an existential formula under a set of constraints
$\Sigma$, then the formula is  \IS-monotonically-determined over $\T$ $\wrt$ $\Sigma$.

\myparagraph{The Entailment Corresponding to the Semantic Property}
As in the previous cases, we instantiate our meta-algorithm by writing out the semantic property as an entailment.

Let $\indomain_\T(x)$ abbreviate the formula:
\[
 \bigvee_{T \in \T } \bigvee_j \exists w_1 \ldots \exists w_{j-1} ~ \exists w_{j+1} \ldots w_{\arity(T)} ~ T(w_1, \ldots w_{j-1}, x, w_{j+1}, \ldots, w_{\arity(T)})
\]
So $\indomain_\T$ states that $x$ is in the domain of a relation in $\T$.
The  entailment  we need is:
\begin{align*}
Q(\vec x) \wedge \Sigma  \wedge \Sigma' \wedge  
\bigwedge_{T \in \T} (\forall \vec y ~ T(\vec y) \rightarrow T'(\vec y)) \wedge
\bigwedge_{T \in \T} (\forall \vec y ~ \bigwedge_i \indomain_\T(y_i) \wedge  T'(\vec y) \rightarrow T(\vec y) ~)\\[-5pt]
\models\; Q'(\vec x)
\end{align*}
Comparing with the two previous entailments, we have the forward implication as before, 
and a restriction of the backward implication.
It is clear that \IS-monotonicity is equivalent to this entailment holding.

\myparagraph{Extracting Interpolants from a Proof of the Entailment: Statement of Result}
We will show later that if we have a ``suitable'' interpolant for  
the entailment corresponding to \IS-monotonicity, then we can extract an existential
reformulation from it, completing another instantiation of the meta-algorithm.

This will give us a proof of another analog of Theorem~\ref{thm:adomcraig}:

\begin{theorem}\label{thm:plt}
If the constraints $\Sigma$ are in equality-free $\rqfo$,
and $\rqfo$ query $Q$ is \IS-monotonically-determined over $\T$, 
then there is an existential first-order formula $\phi(\vec x)$ using only relations in $\T$ such
that $\Sigma \models \forall \vec x ~ Q(\vec x) \leftrightarrow \phi(\vec x)$.  
\end{theorem}

A similar theorem will hold if the constraints are in $\rqfo$
with equality, with the conclusion being
that $\phi$ is an existential formula with inequalities.
There is also an effective variant of the theorem above, an analog
of Theorem \ref{thm:effectivepbdadom}. It states that the existential
$\phi$ can be generated efficiently from a suitable proof of
the entailment for \IS-monotonicity.

We defer the proof of Theorem~\ref{thm:plt} for the moment. 
It will follow from a more general theorem, Theorem~\ref{thm:uspjnegcomplete}, 
proved later.
The reformulation we need will come from applying an interpolation procedure to the entailment above.
\myeat{
\begin{align*}
Q(\vec x) \wedge \Sigma  \wedge \Sigma' \wedge  \\
\bigwedge_{T \in \T} (\forall \vec y ~ T(\vec y) \rightarrow T'(\vec y)) \wedge  \\
\bigwedge_{T \in \T} (\forall \vec y ~ (\bigwedge_i \indomain_\T(y_i) \wedge  T'(\vec y) \rightarrow T(\vec y) ~) \models \\
   Q'(\vec x)
\end{align*}
}
But we need a new interpolation result to guarantee that the interpolant will be existential,
and this extended interpolation theorem is ``Access Interpolation'', described later.

\myparagraph{\lostarski}
The theorem above is closely related to another result in logic, 
the {\lostarski} preservation theorem.
This states that a formula is preserved under extensions exactly when it is equivalent
to an existential formula. 
Roughly speaking, the {\lostarski} theorem is a special case of Theorem~\ref{thm:plt}, 
where $\Sigma$ is empty and $\T$ contains all relations in the schema.
The ``roughly speaking'' disclaimer is because the {\lostarski} theorem deals with classical first-order logic formulas rather than for $\rqfo$, 
and the notion of induced subinstance must be replaced by the analogous notion for structures.

\section{Rewriting with Access Patterns}
Thus far the target of rewriting was specified through vocabulary restrictions.
We wanted a query that used a fixed set of target relations,
perhaps restricted to be positive existential or existential.
We now consider a finer notion of reformulation,
where the target has to satisfy \emph{access restrictions}.

We begin with the definitions of access methods along with a programming language, the RA-plans,
that combines data access by means of a fixed set of access methods with data manipulation using relational algebra queries. 
Recall from Section~\ref{sec:prelims} that relational algebra represents an algebraic programming language that is equivalent in expressiveness
to  first-order logic formulas that are \emph{safe}, in that for every finite instance, the number of satisfying assignments is finite. Implementation of logic-based languages proceeds by  translating a logical
formula into relational algebra.
RA-plans can be thought of as a variation of relational algebra in which we want to abide by a fixed set of data interfaces given
by access methods.  They will thus be the new target of reformulation.

We return to our program of going from semantic properties to rewritings,
following the methodology outlined in the introduction.
We present semantic properties that must hold for a query to be implemented using the interface  given by  a set of access methods, 
and show that these properties can be captured by entailments. 
In the remainder of the section we explain how proofs of these entailments can be converted into plans within our plan languages.
The conversion from proof to plan will again proceed via interpolation. But in this case
we need  a new interpolation theorem tailored to the setting of access methods.

\subsection{Basics of Target Restrictions Based on Access Methods} \label{sec:accessbasics}
We now look at a notion of interface that is closer to the traditional  notion
in programming languages: a set of functions that access the data.
A specification of this interface will  be an extended set of metadata describing both the format of the data (e.g. the vocabulary that would
be used in queries and constraints) and the access methods (functions that interact with the stored data).

An \emph{access schema} consists of:
\begin{itemize}
\item A collection of relations, each of a given arity. 
\item A finite collection $C$ of schema constants ($\smith$, 3, $\ldots$). 
	Schema constants represent a fixed set of values that will be known to a user prior to interacting with the data.
Values that can be used in queries and constraints should be schema constants, as before.
In addition, any fixed values that might be used in plans that implement queries should come from the set of schema constants.
	For example, a plan to answer a query about the mathematics department might involve  first
putting the string ``mathematics'' into a university directory service.
\item For each relation $R$, a collection (possibly empty) of \emph{access methods}\footnote{Our definition of ``access methods'' 
is a variant  of the terminology ``access patterns'' or ``binding patterns'' found in the database literature.}.
Recall that a \emph{position} of a relation $R$ is a number between $1$ and $\arity(R)$.
Each access method $\mt$ is associated with a collection (possibly empty) of positions of $R$ --- the \emph{input positions} of $\mt$.
\item Integrity constraints, which are  sentences in relativized-quantifier first-order logic as before.
\end{itemize}

An \emph{access} (relative to a schema as above) consists of an access method of the schema and a \emph{method binding} ---
a function assigning values to every input position of the method.
If $\mt$ is an access method on relation $R$ with arity $n$,
$I$ is an instance for a schema that includes $R$, and $\accbind$ is a method binding on $\mt$,
then the \emph{output} or \emph{result} of the access $(\mt, \accbind)$ on $I$ is the set of $n$-tuples $\vec t \in \interpret{R}{I}$
such that $\vec t$ restricted to the input positions of $\mt$ is equal to $\accbind$.

An access method may have an empty collection of input positions.
In this case, the only access that can be performed using the method is with the empty method binding.
When a method has no input positions, we say that the access method is ``input-free''.
In Example~\ref{exone} of Section~\ref{sec:textoverview}, the $\univdirectory$ table was assumed
to have such an input-free access method.

The goal is to reformulate source queries in a target language
that represents the kind of restricted computation done over an interface given by an access schema.
We first formalize this operationally, as a language of \emph{plans}. Plans are straight-line programs
that can perform accesses and manipulate the results of accesses using relational algebra expressions.
This language could model, at a high-level,  the plans used internally in 
a database management system. It could also describe the computation done within a data integration system, which might
access remote data via a web form  or web service and then combine data from
different sources  using SQL within its own database management system.

All of our plan languages have as a primitive an \emph{access command}.
Over a schema $\aschema$ with access methods, an access command
 is of the form:
\[
T \Leftarrow_{\outmap} \mt \Leftarrow_{\inmap} E
\]
where: 
\begin{itemize}
\item  $E$ is a relational algebra expression, the \emph{input expression}, 
	over some set of relations not in $\aschema$ (henceforward ``temporary tables'');
\item $\mt$ is a method from $\aschema$ on some relation $R$;
\item $\inmap$, the \emph{input mapping} of the command, is a function from
	the output attributes of $E$ onto the input positions of $\mt$;
\item $T$, the \emph{output table} of the command, is a temporary table;
\item $\outmap$, the \emph{output mapping} of the command, is a bijection from positions of $R$ to attributes of $T$.
\end{itemize}
Note that an access command using an input-free method must take the empty relation algebra expression $\emptyset$ as input.

The manipulation of data retrieved by an access is modeled with the other primitive of our languages,
a  \emph{middleware query command}. These are  of the form $T := Q$, 
where $Q$ is a relational algebra expression over temporary tables and $T$ is a temporary table. 
We use the qualifier ``middleware'' to emphasize
that the queries are performed on temporary relations created
by other commands, rather than on relations of the input schema.

A  \emph{relational algebra-plan} (or simply, \emph{RA-plan}) consists of a sequence of access and middleware query commands, 
ending with  at most one \emph{return command} of the form  $\return~E$, where $E$ is an RA expression.

\begin{myexmp} \label{ex:plan}
We return to Example~\ref{exone} of Section~\ref{sec:textoverview}, 
where we had two sources of information.
One was $\profinfo$,  which was available through an access method $\mt_\profinfo$ requiring input on the first position.
The second was $\univdirectory$, which had an access method $\mt_{\univdirectory}$  requiring no input.
Our query $Q$ asked for ids of faculty named $\smith$.
One plan that is equivalent to $Q$  would be represented as follows 
\begin{align*}
&T_1 \Leftarrow \mt_{\univdirectory} \Leftarrow \emptyset\\
&T_2 := \pi_\eid (\select_{\lname=\smith} T_1)\\
&T_3 \Leftarrow \mt_{\profinfo} \Leftarrow T_2\\
&\return ~ \pi_\eid (T_3)
\end{align*}
Above we have omitted the mappings in writing access commands, since they can be inferred from the context.
We will often do this in plans for brevity.
\end{myexmp}
\myeat{
As syntactic sugar we also allow access commands
$T \Leftarrow_{\outmap} \mt \Leftarrow_{\inmap} E$ where
the number of output attributes of $E$ is strictly greater than
the number of input positions to $\mt$, and the input mapping $\inmap$ is a partial function.
This is just a short-cut for applying a projection on $E$ to the domain of $\inmap$.
}

\myparagraph{Semantics of Plans}
A temporary table is \emph{assigned} in a plan if it occurs on the left side of a command, 
and otherwise is said to be \emph{free}.
The semantics of plans is defined as a function that takes
as input an instance $I$ for $\aschema$ and interpretations of the free tables.
If the plan has no $\return$~statement, the output consists
of interpretations for each assigned temporary table. If the
plan contains a statement $\return~E$, the output is an interpretation of a relation with attributes
for each output attribute of $E$. In the latter case, we refer to this as the \emph{output} of the plan.

An access command $T \Leftarrow_{\outmap} \mt \Leftarrow_{\inmap} E$ is executed by evaluating 
the expression $E$ on $I$ and ``accessing $\mt$ on every result tuple''.
That is, each output tuple of $E$ is mapped to a tuple $t_{j_1} \ldots t_{j_m}$ using the input mapping  $\inmap$.
For each tuple $\vec t=t_1 \ldots t_n \in R$ that ``matches'' (\ie that extends) $t_{j_1} \ldots t_{j_m}$,
$\vec t$ is transformed to a tuple $\vec t'$ using the output mapping $\outmap$. 
The interpretation of $T$ is then the union of all such tuples $\vec t'$.
A middleware query command $T := E$ executes query $E$ on the  contents of the temporary tables mentioned in $E$,
and then assigns the result to temporary table $T$.

A plan is executed by evaluating its commands in sequence, 
with each command operating on the instance formed from the input instance,
adding the interpretations of assigned tables produced by earlier commands.
For a plan having as its final command $\return~E$, 
the output of the plan is the evaluation of $E$ on the instance formed as above. 

We usually assume (without loss of generality) that each table is only assigned once within a plan.
Given plan $\aplan$, temporary table $T$ that occurs in $\aplan$, and an instance $I$ for the schema $\aschema$,
we let $\outcomet{\aplan}{T}{I}$ be the content  of $T$ when $\aplan$ is run on $I$. For plan $\aplan$ including a $\return$ statement,
we let  $\outcomeb{\aplan}{I}$ be the output of $\aplan$ on $I$.
Similarly, given a relational algebra expression $E$ over temporary tables $T_1 \ldots T_n$ of $\aplan$ and instance $I$,
$\outcomet{\aplan}{E}{I}$ represents the result of $E$ when run on $\outcomet{\aplan}{T_1}{I} \ldots \outcomet{\aplan}{T_n}{I}$.

\myparagraph{Fragments of the Plan Language}
We now define fragments of our plan language, analogs
of the fragments of 
relational algebra and first-order logic.
In RA-plans, we allowed arbitrary relational algebra expressions
in both the inputs to access commands and the middleware query commands.
We can similarly talk about \emph{$\spj$-plans}, where the expressions in access and 
middleware query commands are built up from the $\spj$ relational algebra operators
and \emph{$\uspj$-plans} that allow the union operator of relational algebra in addition to $\spj$ operators.
We define \emph{$\uspjneg$-plans} as RA-plans in which relational algebra's difference operator only occurs in a \emph{non-membership check}, 
which tests whether the tuples in a  projection of a temporary table are not in a given relation $R$.
Formally, a non-membership check is a sequence of two commands:
\begin{align*}
T' &\Leftarrow_{\outmap} \mt \Leftarrow_{\inmap} \pi_{a_{j_1} \ldots a_{j_m}}(T)\\
T''&:=T-(T \join T')
\end{align*}

\noindent where in the first command: 
\begin{itemize}
\item $\mt$ is an access method on some relation $R$ with input positions $j_1 \ldots j_m$;
\item the input mapping $\inmap$ maps attribute $a_{j_i}$ to position $j_i$;
\item the attributes of the output table $T'$ are a subset of the attributes of $T$ containing each $a_{j_i}$;
\item the output mapping $\outmap$ maps $j_i$ back to $a_{j_i}$.
\end{itemize}
In the second command, the join condition identifies attributes that have the same name.
$\espj$-plans, $\euspj$-plans, and $\euspjneg$-plans are defined analogously
to the classes above, but not allowing inequality conditions in selections or joins.

\myparagraph{Plans that Answer Queries}
We now define what it means for a plan to correctly implement a query.
Given an access schema $\aschema$, a plan \emph{answers} a query $Q$ \emph{(over all instances)} 
if for every instance $I$ satisfying the constraints of $\aschema$, 
the output of the plan on $I$ is the same as the output of $Q$. 
We often omit the schema from our notation, since it is usually clear from context, 
saying that a plan $\aplan$ answers $Q$.

\myparagraph{From Plans to Specialized Formulas}
We have stated our target language for reformulation as a procedural
language. But to make use of interpolation-based methods, it will be useful to have a logic, consisting of formulas.
We will be able to state an equivalence of plans  in the case of sentences --- that is, Boolean queries. In this case, we can show that  a plan using a given set of access patterns is equivalent to a logical sentence where the quantification structure is compatible with the access patterns. This approach can  be bootstrapped to general logical formulas, but we defer the discussion here, see \cite{interpbook}.

An $\rqfo$ sentence is \emph{executable } (relative to an access schema $\aschema$) if it is built up from equalities
and the formula $\true$ using
arbitrary Boolean operations and the quantifiers:
\begin{align*}
&\forall \vec y ~ [R(\vec x, \vec y)  \rightarrow \phi(\vec x, \vec y, \vec z)] \\[1mm]
&\exists \vec y ~ R(\vec x, \vec y)   \wedge \phi(\vec x, \vec y, \vec z) 
\end{align*}
and for any such quantification above, if $R$ is a $\aschema$ relation,
then  $R$ has an access method $\mt$ such that
 all of the input positions of $\mt$ are occupied by some $x_i$ (that is: by a free variable or constant).

This definition
captures the informal idea that the sentence is compatible
with the access methods:  we should be quantifying only over the output positions, while the values input
positions should be provided.

The next proposition says that executable $\rqfo$ sentences
we can  find a plan that filters an input table down to the subset satisfying the formula.

\begin{proposition} \label{prop:execbool} 
There is a linear time procedure that takes as input an executable $\rqfo$ sentence $\phi$  and produces
an equivalent Boolean  RA-plan $\plan_\phi$.

Furthermore, if the $\rqfo$ sentence is existential with inequalities 
(resp. existential) the result is a $\uspjneg$-plan (resp. $\euspjneg$-plan).
If the sentence is positive existential  with inequalities (resp. positive existential)  the result is a $\uspj$-plan (resp. $\euspj$-plan).

In the other direction, we can convert every Boolean plan (of the given form)
into the corresponding logic.
\end{proposition}

This result is not difficult, and the inductive arguments are given in \cite{interpbook}.
It will allow us to deal with logical formulas having these ``access-restricted quantifiers'' from now on.

\myparagraph{The Accessible Part}
In order to begin our instantiation of the meta-algorithm for plans, we need a semantic property
corresponding to existence of a plan. The property should say that  formula   ``only depends on the  accessible data'', and thus we need to formalize what accessible data means.

\begin{definition}\label{def:accpart}
Given an instance $\aninst$ for schema $\aschema$ the \emph{accessible part of $\aninst$}, denoted $\accpart(\aninst)$, and the \emph{accessible values of $\aninst$},
denoted  $\accessible(\aninst)$. Informally the accessible part consists of all the facts over $\aninst$ that can be obtained by starting with empty relations and iteratively entering values into the access methods.
If $\aschema$ contains no schema constants, this will be 
 an instance containing a set of facts $\acc{R}(v_1 \ldots v_n)$,
where $R$ is a relation and  $v_1 \ldots v_n$ are values in the domain of $\aninst$ such that $R(v_1 \ldots v_n)$ holds in $\aninst$,
obtained by starting with relations $\acc{R}_0$  and $\accessible_0$
empty
and then iterating the following process until a fixpoint is reached:

\begin{align*}
    \accessible_{i+1} &= \accessible_{i} ~ \cup \mathop{\bigcup_{\text{$R$ a relation}}}_{j \leq \arity(R)} \pi_j (\acc{R}_i)
\end{align*}
and
\begin{align*}
	\acc{R}_{i+1} &= \acc{R}_i ~ \cup \hspace{-65pt}
		\mathop{\bigcup_{(R,\{j_1, \ldots, j_m\})}}_{\text{there is a method of $\aschema$ on $R$ with inputs $j_1, \ldots, j_m$}} \hspace{-75pt}
	\{(v_1 \ldots v_n) \in \interpret{R}{I}~|~v_{j_1} \ldots v_{j_m} \in \accessible_i \}
\end{align*}

%

Above $\pi_j(\acc{R}_i)$ denotes projection of $\acc{R}_i$ on the $j^{th}$ position.
For a finite instance, this induction  will reach a fixpoint after $|\aninst|$ iterations, 
where $|\aninst|$ denotes the number of facts in $\aninst$.
For an arbitrary instance the union of these instances over all $i$ will be a fixpoint.
Assuming $\aschema$ does include schema constants, we modify the definition
by starting with an $\accessible_0$ consisting of the schema constants, rather than being
empty.
\end{definition}
Above we consider $\accpart(I)$ as a database instance for the schema
with relations $\accessible$ and $\acc{R}$. The accessible
values are the union of the $\accessible_i$.

In the case of vocabulary-based access-restrictions, the accessible part of an instance
just represents the restriction of  the instance to the relations in the subsignature $\T$.
Thus access determinacy of a query $Q$ in the case of vocabulary-based restrictions
is the same as determinacy of $Q$ with respect to the subsignature.

\myparagraph{The Semantic Property and Entailment for RA-Plans}
We now give the analogous property and entailment for RA-plans.
$Q$ is said to be \emph{access-determined} over $\aschema$ if for all instances $\aninst$ and $\aninst'$ satisfying the constraints
of $\aschema$ with $\accpart(I)=\accpart(I')$ we have $\outcomeb{Q}{I}=\outcomeb{Q}{I'}$.
If a query is \emph{not} access-determined, it is obvious that it cannot be answered through any plan,
since it is easy to see that any plan can only read tuples in the
accessible part.

\begin{myexmp} \label{ex:accessd}
We return to the setting of Example~\ref{exone}, where  we have a $\profinfo$ table available via a web form, containing information about faculty,
including their last names, office number, and employee id, but with only an access method $\mt_{\profinfo}$ that
requires giving an employee id as an input. 
We consider a query $Q$ asking for ids of faculty named $\smith$, where
$\smith$ is a schema constant.

We show that $Q$ is not access-determined.
For this, take  $I$ to be any instance that contains  exactly one tuple,
 with $\lastname$ $\smith$, but with no schema constant as its employee id.
Let $I'$  be the empty instance.
The accessible parts of $I$ and $I'$ are empty, since in both cases when we enter all the constants we know about in $\mt_\profinfo$,
we get the empty response.
But $Q$ has an  output on $I$ but no output on $I'$.

$I$ and $I'$ witness that $Q$ is not access-determined. From this we see that $Q$ can not be implemented by any plan using 
$\mt_\profinfo$.
\end{myexmp}

We now turn to the corresponding entailment for access-determinacy. The accessible part itself is not a first-order definable object --- it requires an infinite disjunction or recursion. But it will turn out that ``agreeing on the accessible part'' can be expressed in a first-order way.

We start with a schema $\aschema$ consisting of relation symbols and RQFO sentences --- integrity constraints, as well as access restrictions. The \emph{bidirectional accessible schema for $\aschema$}, 
denoted $\accsb(\aschema)$, is a vocabulary without access restrictions, and also a set of integrity constraints in the form of RQFO sentences. It is defined as follows:
\begin{itemize}
\item The constants are those of $\aschema$.
\item The relations are those of $\aschema$, a unary relation $\accessible(x)$ (``$x$ is an accessible value'') 
	plus a copy of each relation $R$ of $\aschema$ called $\infacc{R}$ (the ``inferred accessible version of $R$''). 
\item The constraints are those of $\aschema$ (referred to as ``$\aschema$ constraints'' below) 
    for each access method $\mt$ on relation $R$ of arity $n$ with input positions $j_1 \ldots j_m$
	along with  the following constraints (dropping universal quantifiers on the outside for brevity)
\begin{itemize}

	\item \emph{forward accessibility axioms}:  we have a rule:
		\begin{eqnarray*}
		\accessible(x_{j_{1}}) \wedge \ldots \wedge \accessible(x_{j_{m}}) \wedge R(x_1 \ldots x_n)
		\rightarrow\\ 
		\infacc{R}(x_1 \ldots x_n) \wedge \bigwedge_j \accessible(x_j)
		\end{eqnarray*}
        \item \emph{backward accessibility axioms}
        \begin{eqnarray*}
\bigwedge_{i \leq m} \accessible(x_{j_{i}}) \wedge  \infacc{R}(x_1 \ldots x_n) 
\rightarrow 
R(x_1 \ldots x_n) \wedge \bigwedge_{i} \accessible(x_i)
\end{eqnarray*}
		In addition, we have $\accessible(c)$ for each constant $c$ of $\aschema$.

	\item A copy of each of the original integrity constraints, 
		with each relation $R$ replaced by $\infacc{R}$, denoted ``$\infacccon$ constraints'' below.
\end{itemize}
\end{itemize}

\medskip\par\noindent

Given a query $Q$, its \emph{inferred accessible version} $\accq{Q}$ is obtained by replacing each relation $R$ by $\infacc{R}$.
Informally, $\accq{Q}$ represents the fact that the existence of a witness to $Q$ can be obtained through making
accesses and reasoning.

We  overload $\accs(\aschema)$ to also refer to the conjunction of axioms in this schema $\accs(\aschema)$.

A way of motivating the axioms is to think informally of the relations $R$ $\infacc{R}$ 
as representing the instances considered in access determinacy.
By having two copies of the constraints, we are stating that both instances satisfy the constraints of the schema.
The forward and backward accessibility axioms will ensure that
these two instances have the same accessible part.
The entailment thus captures that if two instances have the same accessible part, and both satisfy the constrains, and one satisfies
query $Q$, then so does the other.

Again, we show that this entailment captures the proposed preservation property, access-determinacy.
\begin{claim} \label{clm:raanddeterminacy}
The following are equivalent (for any Boolean $\rqfo$ query $Q$ and
access schema consisting of $\rqfo$ constraints):
\begin{enumerate}
\item $Q$ entails  $\accq{Q}$ with respect to the rules in  $\accsb(\aschema)$
\item $Q$ is access-determined over $\aschema$
\end{enumerate}
\end{claim}

\begin{proof}
We prove that the first item implies the second.
Fix $\aninst$ and $\aninst'$ satisfying the schema with the same accessible part, and assume $\aninst$ satisfies $Q$.
Consider the instance $\aninst''$ for $\accsb(\aschema)$ formed by interpreting the relations $R$ as in $\aninst$,
the relation $\accessible$ by the accessible values of $\aninst$, and each $\infacc{R}$ by the interpretation of $R$ in $\aninst'$.
Then one can  verify that $\aninst''$ satisfies the constraints of $\accsb(\aschema)$.
Since $\aninst$ (and hence $\aninst''$) satisfies $Q$, and we are assuming that $Q$ entails $\accq{Q}$
with respect to $\accsb(\aschema)$ we can conclude that $I''$ must satisfy $\accq{Q}$.
So $Q$ holds in $\aninst'$ as required.

We complete the proof of the claim by arguing from the second item to the first.
Suppose $Q$ is not contained in $\accq{Q}$ with respect to the rules in $\accsb(\aschema)$.
Hence there is an instance $I^{\accsb}$ satisfying the rules of $\accsb(\aschema)$ and also satisfying $Q \wedge \neg \accq{Q}$.
Let $I_1$ consist of the restriction of $I^{\accsb}$ to the original schema relations.
Let $I_2$ consist of the inferred accessible relations from $I^{\accsb}$, renamed to the original schema.

We claim that a fact $R(e_1 \ldots e_n)$ in the  accessible part of $I_1$ is in the accessible part of $I_2$.
This is proven by induction on the point in which the fact goes into the accessible part, using the forward accessibility axioms.
Arguing symmetrically (now using the backward axioms), we find that $I_1$ and $I_2$ have the same accessible part, and hence they contradict access-determinacy.
\end{proof}

We are now ready to state our main results on the relationship between semantic properties, entailments, and plans.

\begin{theorem} \label{thm:racomplete}
For any Boolean $\rqfo$ query $Q$ and access schema $\aschema$ 
with constraints specified in $\rqfo$,
there is an RA-plan answering $Q$ (over instances satisfying $\aschema$) if and only if $Q \wedge \accsb(\aschema) \models \accq{Q}$.
If the query and constraints do not include equality, 
then the RA-plan will not make use of equality in any of its RA expressions.

Further, from any tableau proof witnessing $Q \wedge \accsb(\aschema) \models \accq{Q}$ we can extract (in linear time) an RA-plan for $Q$ over $\aschema$. 

\end{theorem}

\noindent Using Claim~\ref{clm:raanddeterminacy}, we can restate Theorem~\ref{thm:racomplete}:

\begin{quote}
For any Boolean $\rqfo$ query $Q$ and access schema $\aschema$ with constraints specified in $\rqfo$,
there is an  RA-plan answering $Q$ (over instances of $\aschema$) 
if and only if $Q$ entails $\accq{Q}$ with respect to the rules in $\accsb(\aschema)$
if and only if $Q$ is access-determined.
\end{quote}

\noindent In the direction from right to left we are again going from a preservation property
to a syntactic restriction. Thus Theorem~\ref{thm:racomplete} can be considered as an ``access-restricted
variant'' of the definability theorems given in the vocabulary-based case.

\subsection{The Interpolation Result Behind the Rewriting Results: Access Interpolation}
We have instantiated the first steps of our meta-algorithm for reformulation via relational algebra access plans. 
For RA-plans, 
we have stated a semantic property of a query that is required for it to have an equivalent target RA plan.
And we have shown that the semantic property is equivalent to an entailment.
The next ingredient in proving  the main results relating plans to  proofs of these entailments
is an interpolation theorem that tracks the ``access methods used in the interpolant''.

\myparagraph{Interpolation and Access Methods}
Recall Proposition~\ref{prop:execbool}. 
It states that Boolean RA-plans have the same expressiveness as $\rqfo$
sentences that are executable for membership checks.
For simplicity, we call these ``executable FO Boolean queries'' below.

By this corollary, we can prove our main results by finding reformulations that are executable FO Boolean queries, rather than RA-plans.
To find these reformulations, we require a version of Craig interpolation that allows us to relate the ``binding patterns'' ---
\ie ~ the subset of variables that are quantified in each relativized-quantifier quantification ---  
used in the interpolant $\chi$ of the entailment $\lambda\models \rho$ with those used in $\lambda$ or $\rho$.
When we apply this theorem to the entailment of $\accq{Q}$ by $Q$, we can
conclude that the interpolant is an executable FO Boolean query.

We associate to $\rqfo$ formulas the set of binding patterns used in quantification, 
where a binding pattern is a relation and a subset of the positions.
This is done by structural induction on the formula, assuming that the formula is first put into negation normal form.

\[\begin{array}{@{}lll@{}}
\bindpat(\top) ~=~ \bindpat(x=y) &=& \emptyset \\
\bindpat(R(t_1, \ldots, t_n)) &=& \{ (R,\{1,\ldots, n\}) \} \\
\bindpat(\neg\phi) &=& \bindpat(\phi) \\
\bindpat(\phi\land\psi) &=& \bindpat(\phi)\cup \bindpat(\psi)  \\
\bindpat(\phi\lor\psi) &=& \bindpat(\phi)\cup \bindpat(\psi) \\
\bindpat(\exists\vec{x} ~ (R(t_1, \ldots, t_n)\land \phi)) &=& \bindpat(\phi)\cup\{(R,\{i\mid t_i\not\in\vec{x}\})\} \\
\bindpat(\forall\vec{x} ~ (R(t_1, \ldots, t_n)\to \phi)) &=& \bindpat(\phi)\cup\{(R,\{i \mid t_i\not\in\vec{x}\})\}
\end{array}\]
For example,
\[
\bindpat(\exists x ~\exists y ~ (R(x,y)\land \forall z ~ (S(x,y,z)\to U(x,y,z)))) = 
\{(R,\emptyset), (S,\{1,2\}), (U,\{1,2,3\})\}
\]
We can similarly talk about the \emph{universal binding patterns} of a formula  --- those
that arise from a subformula $\forall\vec{x} ~ (R(t_1, \ldots, t_n)\to \phi)$ --- and the \emph{existential binding
patterns} --- those that arise from a subformula $\exists\vec{x} ~ (R(t_1, \ldots, t_n)\land \phi)$. These can be defined
formally by modifying the inductive definition above. Note that every binding pattern of a formula is either a universal
pattern or an existential pattern, since a non-quantified relation can
be considered as a vacuous case of existential quantification
(if occurring positively) or universal quantification (if occurring
negatively).

Intuitively, $\bindpat(\phi)$ describes the kind of access that is used if $\phi$ is
evaluated in an instance using a straightforward inductive evaluation procedure.
For a sentence $\phi$, if  for each pattern $(R  \{s_1 \ldots s_j\})$ in $\bindpat(\phi)$  $\aschema$  contains an access method
on $R$ whose input positions are contained in  $\{s_1 \ldots s_j\}$, then $\phi$ is an executable FO Boolean query.
We say that a binding pattern $(R_1, \{s_1 \ldots s_j\})$  is \emph{covered by} another binding pattern $(R_2, \{t_1 \ldots t_k\})$
if $R_1=R_2$ and $\{s_1 \ldots s_j\}$ is a superset of $\{t_1 \ldots t_k\}$.   To show that a sentence $\phi_1$
is an executable FO Boolean query, it suffices to get another executable FO Boolean query $\phi_2$ whose binding patterns
cover those of $\phi_2$.

Recall that a relation  $R$ \emph{occurs positively}  (\emph{negatively}) in a formula $\phi$ if some occurrence
of $R$ in $\phi$ is in the scope of an even (odd) number of negations. For the purpose of this definition,
we view the implication symbol as a shorthand: $\psi\to\chi$ stands for $\neg\psi\lor\chi$. For example, in the formula $\forall x ~(P(x)\to \exists y ~ R(x,y))$, the relation 
$P$ occurs negatively and the relation  $R$ occurs positively.

\begin{theorem}[Access interpolation] \label{thm:access-interpolation}
Let $\lambda$ and $\rho$ be $\rqfo$ sentences such that $\lambda \models \rho$.
Then there exists an $\rqfo$ sentence $\chi$ such that
\begin{enumerate}
\item $\lambda \models \chi$ and $\chi \models \rho$.
\item A relation  occurs positively (negatively) in $\chi$ only if it occurs
    positively (negatively) in both $\lambda$ and $\rho$.
\item A constant  occurs in $\chi$ only if it occurs both in $\lambda$ and $\rho$.
\item Every existential binding pattern of $\chi$ is covered by an existential binding pattern
	of $\rho$, and the relation it binds occurs positively in $\lambda$.
	Every universal binding pattern of $\chi$ is covered by a universal binding pattern of $\lambda$, 
	and the relation it binds occurs negatively in $\rho$.


\item If $\lambda$ and $\rho$ are both equality-free, then $\chi$ is equality-free.
\end{enumerate}
Furthermore, $\chi$ can be computed in polynomial time from a proof (in a suitable proof system) of the entailment $\lambda \models \rho$.
\end{theorem}

We show how Theorem~\ref{thm:access-interpolation} suffices to prove  Theorem~\ref{thm:racomplete}. Recall the theorem:

\begin{quote}
For any Boolean $\rqfo$ query $Q$ and access schema $\aschema$ containing 
constraints expressible in $\rqfo$,
there is an RA-plan answering $Q$ (over instances in $\aschema$) 
if and only if $Q$ entails  $\accq{Q}$ with respect to the
rules in  $\accsb(\aschema)$.
If the query and constraints do not include equality, 
then the RA-plan will not make use of equality in any of its relational algebra expressions.
\end{quote}
From Claim~\ref{clm:raanddeterminacy}, we obtain  the ``plan-to-proof'' direction of Theorem~\ref{thm:racomplete}. 
Suppose $Q$ does not imply $\accq{Q}$ with respect to the rules in $\accsb(\aschema)$.
By Claim~\ref{clm:raanddeterminacy}, $Q$ is not access-determined, and it follows that no plan can answer $Q$.

For the ``proof-to-plan'' direction of Theorem~\ref{thm:racomplete},
we assume $Q$ entails $\accq{Q}$ and construct an RA-plan that answers $Q$.
We will use a slight modification of the axiom schema $\accsb$, 
denoted $\accsbp$ in which the relation $\accessible$ does not appear.
In every forward accessibility axiom an atom $\accessible(x)$ on the left is replaced by a relation $\infacc{R}(\vec z)$, 
where $\vec z$ contains $x$ in at least one position (and the  other variables are universally quantified). 
Occurrences of $\accessible$ on the right are dropped.
For example, the axiom $\accessible(x) \wedge R(x,y) \rightarrow \infacc{R}(x,y) \wedge \accessible(y)$ would be replaced by many axioms, 
including $\infacc{S}(x,w,z) \wedge R(x,y) \rightarrow \infacc{R}(x,y)$.
We also allow all  variants of these axioms in which  free variables corresponding to input
positions on the left are substituted by schema constants.
Similarly, in every backward accessibility axiom we replace $\accessible(x)$ on the left by an atom in the original schema containing $x$, 
again dropping occurrences of $\accessible$ on the right.

\begin{proposition} \label{prop:axequiv}
$Q$ proves $\accq{Q}$ using the axioms of $\accsb$ if and only if $Q$ proves  $\accq(Q)$ in the modified schema $\accsbp$.
\end{proposition}
\begin{proof}
In one direction, suppose that $Q$ does not prove $\accq{Q}$ using the axioms of $\accsb$. 
Then by Claim~\ref{clm:raanddeterminacy}, $Q$ is not access-determined.  
Fix $\inst$ and $\inst'$ instances for the original schema satisfying $\Sigma$ with the same accessible part, 
but  disagreeing on the output of $Q$.  Without loss of generality, we can assume
that there is a tuple $\vec d$ in the output of $Q$ on $\inst$ but not in the output of  $Q$ on $\inst'$.
By applying an isomorphism to non-accessible values of $\inst'$ and $\inst$, 
we can assume that every non-accessible value of $\inst'$ is not in $\inst$ and vice versa.
Let $\inst^*$ be the instance   in which relations $R$ are interpreted
as in $\inst$ and relations $\infacc{R}$ are interpreted, as in $\inst'$, by
a new relation symbol.

We argue that $\inst^*$ satisfies $\accsbp$. 
Clearly both the original relations and the relations $\infacc{R}$ satisfy $\Sigma$. 
Consider a modified forward accessibility axiom (universal quantifiers omitted, and with no schema
constants for simplicity):
\[
\infacc{R_{j_1}}( \ldots x_{j_1} \ldots ) \wedge \ldots \infacc{R_{j_m}}(\ldots x_{j_m} \ldots ) \wedge R(\vec x) \rightarrow \infacc{R}(\vec x)
\]
where $R$ has an access method on positions $j_1 \ldots j_m$.
Suppose we have a tuple $c_{1} \ldots c_m$ in $\inst^*$ satisfying
the left-hand side of this implication. Then $c_{j_1} \ldots c_{j_m}$
must be accessible values of $\inst$ and $\inst'$. Since the fact $R(\vec c)$ holds
in $\inst$, it must be in the accessible part of $\inst$, and hence in the
accessible part of $\inst'$. So $\infacc{R}(\vec c)$ holds in $\inst^*$ 
as required. The backward accessibility axioms are argued symmetrically.

The tuple $\vec d$ is in $\interpret{Q}{\inst^*}$,
since it is in $\interpret{Q}{\inst}$ and $Q$ is a formula using  the relations in the original schema.
But $\vec d$ is not in $\interpret{Q}{\inst'}$.
Therefore $\vec d$  can not be returned by $\accq{Q}$ on $\inst^*$.
So $\inst^*$  witnesses that $Q$ does not prove $\accq{Q}$ in $\accsbp$.

In the other direction, suppose we have $\inst^*$ witnessing that $Q$ does not prove $\accq(Q)$ in $\accsbp$.
Expand $\inst^*$ to an instance $\inst^+$ for the signature extended with $\accessible$, 
by interpreting $\accessible$ by all values of schema constants
unioned with all values that lie in the domain of some $\infacc{R}$ relation
and in the domain of some relation of the original schema.
We show that the resulting instance satisfies the constraints of $\accsb$. 
The accessibility axioms follow directly from the corresponding axioms of $\accsbp$.
Similarly, the output of $Q$ in $\inst^+$ is the same as the output of $Q$ in $\inst$, 
while the output of $\accq{Q}$ on $\inst^+$ is the same  as the output of $\accq{Q}$ on $\inst^*$.
Therefore  $Q$ does not prove $\accq{Q}$ in $\accsb$.
\end{proof}

We can rephrase the modified assumption on $Q$ as:
\[
Q \wedge \Sigma_1 \models  \Sigma_2 \rightarrow \accq{Q}
\]
where $\Sigma_1$ contains the $\aschema$ constraints
along with the ``backwards'' accessibility axioms going from $\infacc{R}$ relations to $R$ relations within $\accsbp$,  
while $\Sigma_2$ contains the forward accessibility axioms and  the $\infacccon$  constraints. 

By the access interpolation theorem there is an $\rqfo$ formula $\chi$ in negation normal form such that:
\begin{enumerate}
\item $Q \wedge \Sigma_1 \models \chi$ 
\item $\chi \models \Sigma_2 \rightarrow  \accq{Q}$
\item a relation occurs positively (respectively negatively) in $\chi$ if
and only if it occurs positively 
      (resp. negatively) on both sides of the entailment.
\item for a binding pattern $p$ on a relation $U$ in $\chi$ if $p$ is  existential, 
      then it is covered by an existential pattern in $\Sigma_2 \rightarrow \accq{Q}$, and $U$  occurs positively in $Q \wedge \Sigma_1$.
If $p$ is a universal pattern, then it is covered by a universal pattern
of $Q \wedge \Sigma_1$ and $U$ occurs negatively in $\Sigma_2 \rightarrow \accq{Q}$.
\end{enumerate}

Because we are dealing with the modified axioms $\accsbp$, the only
relations in the entailment  are the original schema relations and their $\infacc$ copies.
Occurrences of the original schema relations $R$ on the right-hand side of the entailment are all in the forward accessibility axioms, 
and they correspond exactly to access methods of the schema. 
It follows that every existential pattern over the relations $R$ in $\chi$ is of the proper form.
Furthermore the relations $R$ occur only positively on the right ---
since $\Sigma_2 \rightarrow \accq{Q}$ in negation normal form is $(\nneg \Sigma_2) \vee \accq{Q}$,
and $\nneg$ applied to a forward accessibility axiom is of the form 
$\exists \vec x ~ R(\vec x) \wedge \ldots \wedge \neg \infacc{R}(\vec x)$.
So relations $R$ can occur only positively in $\chi$, and hence can occur only in existential patterns of $\chi$.
We conclude that all the patterns in $\chi$ involving the relations of the original schema must be existential and of the proper form.

Now let us  examine occurrences of relations of the form $\infacc{R}$ within $\chi$.
The last property of $\chi$  above implies that  any existential pattern on relations of the form 
$\infacc{R}$ would need to correspond to a positive occurrence of $\infacc{R}$ in $Q \wedge \Sigma_1$. 
But there are no such occurrences.
Any universal pattern on relations of the form $\infacc{R}$ would need to correspond to a universal pattern in $Q \wedge \Sigma_1$, 
and hence must be covered by the use of $\infacc{R}$ in a backward axiom. 
Hence all such occurrences are covered by an access method of the schema.

We  conclude that $\chi$ is an executable FO Boolean query. 
Let $\deacc(\chi)$ be the result of changing all relations of the form $\infacc{R}$ in $\chi$ to $R$.
We claim that $\deacc(\chi)$ is an executable rewriting of $Q$ 
(and hence can be converted to an RA-plan using Proposition~\ref{prop:execbool}).

We justify this by proving containments between $Q$ and $\deacc(\chi)$ in both directions.
Suppose tuple $\vec t$ is returned by $Q$ on an instance $\inst$ of $\aschema$.
Let $\inst'$ be the instance for the augmented schema in which both $\infacc{R}$ and $R$ relations are interpreted as in $\inst$.
We can see that $\inst'$ satisfies the constraints in $\Sigma_1$ and $\Sigma_2$.

Applying the first condition on an interpolant above, we infer that $\chi$ holds of $\vec t$ in $\inst'$, which
implies that $\chi$ holds of $\vec t$ in $\inst$. 
But $\deacc(\chi)$ evaluated on $\inst$ yields the same set of tuples as evaluating $\chi$ on $\inst'$. 
So $\vec t$ satisfies $\deacc(\chi)$ in $\inst$.

In the other direction, suppose  tuple $\vec t$  satisfies $\deacc(\chi)$  in $\inst$.
Then letting $\inst'$ be as above, we have that $\vec t$ satisfies $\chi$ in  $\inst'$.
By the second property of an interpolant we have that  $\Sigma_2 \rightarrow \accq{Q}$ holds in $\inst'$. 
Since $\inst'$ satisfies $\Sigma_2$,  $\vec t$ is returned by $\accq{Q}$  in $\inst'$. 
This tells us that  $\vec t$ is returned by $Q$ as required.

This completes the proof of the first assertion in Theorem~\ref{thm:racomplete}. 
The assertion about equality-free queries and constraints follows since the Access interpolation theorem produces
an equality-free interpolant when the entailment involves equality-free formulas, and the other transformations
(\eg, from nested plans to RA-plans) do not introduce equality.

\subsection{Variants of the Results for $\uspj$-Plans}
We now state an analogous  result  for negation-free plans, 
which will be an access-related variant of the Projective Monotone Preservation Theorem, 
Theorem~\ref{thm:plp}. 

Here we use the schema $\accs(\aschema)$ obtained from $\accsb$ by
throwing away all the backward accessibility axioms, and keeping only the forward ones.

\begin{theorem} \label{thm:spjcomplete}
For any Boolean $\rqfo$ query $Q$ and 
access schema $\aschema$ containing constraints specified in $\rqfo$,
the following are equivalent:
\begin{itemize}
\item there is a $\uspj$-plan answering $Q$ (over instances in $\aschema$) 
\item $Q$ entails $\accq{Q}$ with respect to $\accs(\aschema)$ ($Q \wedge \accs(\aschema) \models \accq{Q}$).
\end{itemize}

Furthermore, for every tableau proof witnessing $Q \wedge \accs(\aschema) \models \accq{Q}$, we can extract a $\uspj$-plan.
If the query and the constraints of $\aschema$ are specified by $\rqfo$ formulas without any equalities (\eg, TGDs), 
then we can replace $\uspj$-plan with $\EUSPJ$-plan in the above statement. 
So we do not need inequalities in the reformulation unless we have equalities in the constraints.
\end{theorem}

\myparagraph{The Schema $\accsneg$ and $\uspjneg$-Plans}
We now state an extension of the ``Projective {\lostarski} theorem'', Theorem~\ref{thm:plt}, to
the setting of access methods, again focusing on the case of Boolean queries.

\begin{theorem} \label{thm:uspjnegcomplete}
For any Boolean $\rqfo$ query $Q$ and access schema $\aschema$ containing constraints specified in $\rqfo$,
the following are equivalent:
\begin{itemize}
\item there is a $\uspjneg$-plan answering $Q$ (over instances for $\aschema$)
\item $Q \wedge \accsneg(\aschema) \models \accq{Q}$. That is, $Q$ entails  $\accq{Q}$ with respect to $\accsneg(\aschema)$.
\end{itemize}

Furthermore, from every tableau proof witnessing  $Q \wedge \accsneg(\aschema) \models \accq{Q}$,
we can effectively extract a $\uspjneg$-plan.
If the constraints of $\aschema$ and the query $Q$ are specified by equality-free $\rqfo$ sentences (in particular, TGDs), 
then we can replace $\uspjneg$-plan by  {\EUSPJneg}-plan.
\end{theorem}

In the vocabulary-based setting, the entailment $Q \wedge \accsneg(\aschema) \models \accq{Q}$ reduces to 
the entailment given in Section~\ref{sec:vocabexistential}, 
which was shown there to be equivalent to the query being \IS-monotonically-determined. 
And clearly, in the vocabulary-based setting, a $\uspjneg$-plan can be expressed as an $\existsineq$  formula.
Thus Theorem~\ref{thm:uspjnegcomplete} generalizes the ``Projective {\lostarski} Theorem'',
Theorem~\ref{thm:plt}, characterizing queries that can be reformulated using $\existsineq$ formulas.

Both of these results  are proven in the same way as the result for RA plans. By Proposition \ref{prop:execbool}, instead of dealing
with plans as the target language, we deal with RQFO formulas with
certain binding patterns. We apply the Access Interpolation Theorem, and then analyze
the binding patterns in the resulting formulas, and show that we satisfy the requirements to convert to the associated plan.

\section{State of Play and Future Directions} \label{sec:future}
Here we have given a flavor of the use of interpolation in database rewriting problems, focusing on the cases of vocabulary-based restriction and access patterns.

Several other directions have been explored in the literature:

\myparagraph{Richer Data Models} The examples we have gone through here both concern relational data - i.e. traditional set-based semantics. In databases there are tools and query languages for other collection types, like bags, nested sets, nested bags, sequences, and nested sequences. The main work we know of beyond relational data is about nested sets, summarized in \cite{nestedlmcs}. The model starts with an infinite scalar type, and builds up objects via iterating tupling and powersets: thus one can have sets of sets of pairs of scalars. Pairs of sets of sets of scalars, etc.  There is a standard query language in the literature for this data model, nested relational calculus (NRC). In \cite{nestedlmcs} a notion of implicit definability/determinacy is defined for first-order logic formulas in a simple set-theoretic language. And it shown that any query over nested sets that is determined in this sense can be converted to NRC, effectively from a certain kind of proof.
Although at first glance this appears to be a version of Beth Definability for a fragment of set theory, in fact the only set theoretic axiom used is extensionality: sets with equal members are equal. Thus implicit definability in the context of nested data is really
a variation of standard implicit definability where the uniqueness is only ``up to extensional equivalence''. Extensional equivalence is definable in the language of set theory. Thus the question is whether there is a version of Beth's theorem that considers uniqueness up definable equivalence. One of the key insights in \cite{nestedlmcs} is that there is a variation of Beth that applies to this setting,  originating in work of Gaifman \cite{gaifman} under the name of ``rigid categoricity''.  While the results on rigid categoricity were proven model theoretically in \cite{gaifman} (see \cite{hodgesbook}), in \cite{nestedlmcs} a proof-theoretic argument is developed for the special case of extensional-equivalence.

A natural question is what about the case of \emph{bags}: finite multi-sets. While there is no standard semantics for first-order logic over bags, there is a semantics for conjunctive queries, and thus one can talk about determinacy of CQ views over bag semantics.
Note that even for Boolean queries, it is not clear whether determinacy is decidable, or whether it implies rewritability in some natural query language for bags.
The main results on this topic are in \cite{jerzybag}, where the case of Boolean views is demystified.

\myparagraph{Richer Query Languages over Traditional Data Models} The work we overview deals with fragments of first-order logic. The only other work we are aware of in this space
concerns fragments of the recursive query language Datalog, a fragment of fixpoint logic.
The first work in this space was \cite{regpathdet}, which investigates the case of views in
a fragment of Datalog, the ``regular path queries'': it is shown that determinacy implies rewritability, an analog of implicit definability.
The case of general Datalog queries and views was investigated in \cite{mondetdatalogj,monotonekr}.  The only implicit definability results are for fragments of Datalog where containment is decidable --- which means, very limited fragments. The main technique used is a variation of uniform interpolation, which hold for limited fixpoint languages \cite{mucalcuniforminterp}. See \refchapter{chapter:uniform} for a discussion of uniform interpolation algorithms.

\myparagraph{Other Proof Systems and Implementations}
Toman and Weddell were the first to develop interpolation-methods for query rewriting. Their initial approach, based on a tableaux calculus,  is presented in \cite{tomanweddell}, with follow-up work in
\cite{tomanbeth,tomanmagic}.  A rewriting approach for access patterns was developed
in \cite{BenediktLT15, usvldb14}, but only in a restricted context where interpolation is not necessary. 
\cite{qreformijcai} developed an approach for interpolation-based rewriting where the underlying proof system is resolution. This made use of the standard algorithmic approach to resolution-based interpolation due to Huang \cite{huang}, which is discussed in \refchapter{chapter:automated}.

\section{Bibliographic Remarks}
The precursors of the interpolation-based approach to rewriting are
in preservation or ``semantics to syntax'' theorems in model theory, such as the \lostarski ~ theorem, which states that preservation under extensions is equivalent to being rewritable as an existential formula: see for example \cite{hodgesbook, ChangKeisler}.
After the work of Lyndon \cite{lyndon59}, these results would be stated and proved
model-theoretically, and effective aspects would be ignored. In addition,
the presentations did not include the ``projective versions'', in which
one looks for a restricted formula in a particular vocabulary: these variants are important for the applications in databases.

The idea that Craig interpolation is relevant to query reformulation in databases, from the point of view of theory, dates
from the work of Segoufin and Vianu \cite{SVconf}, later explored jointly with Nash \cite{NSV}. Their investigation focused on the case of views --- a special case of vocabulary-based reformulation. They phrase their results as showing that relational algebra was sufficient in seeking a reformulation over views. 

Toman and Weddell \cite{tomanweddell} were the first to explore the use of effective interpolation to actually generate rewritings, in the context of query reformulation with constraints. Toman and Weddell also note connections with other query reformulation methods that have an interpolation-like feel, like the `chase and backchase method'', which is used to reformulate queries in the presence of a certain class of background theories, so-called Tuple-Generating Dependencies (TGDs): for an overview of the chase and backchase, see \cite{candbsigrecord}.

\pagebreak
Much of the material in this chapter here is taken from \cite{interpbook}, which extends earlier work \cite{uspods,ustods} looking systematically at interpolation for both vocabulary-based reformulation and rewriting with access methods.
In this chapter we have given some samples of vocabulary-based reformulation for ``first-order'' (or relational algebra)
queries, and on access-method based reformulation for the same queries, but the book covers several other cases. The book and the related papers explicitly connect these reformulation methods with preservation theorems and other semantics-to-syntax results in the earlier model theory literature.

The case of so-called ``nested relations'', which can be considered as a fragment of set theory,  is taken from \cite{nestedlmcs}. We have overviewed these results briefly in Section \ref{sec:future}.

The case of Datalog (fixpoint) queries and views is considered in \cite{regpathdet,mondetdatalogj,monotonekr}. Here the results are related to uniform interpolation, and they work via automata-theoretic methods, as in \cite{mucalcuniforminterp}.

\bibliography{algs,taci}

\end{document}